\documentclass[letter,11pt]{article}
\usepackage{amssymb, amsmath, amsthm, mathtools, pdflscape, subcaption, comment}
\usepackage[margin=1in]{geometry}
\allowdisplaybreaks[1]
\usepackage[longnamesfirst, round]{natbib}
\usepackage[colorlinks=true, linkcolor=blue, filecolor=blue, urlcolor=blue, citecolor=blue]{hyperref}
\usepackage{graphicx}

\newtheorem{assumption}{Assumption}
\newtheorem{corollary}{Corollary}
\newtheorem{theorem}{Theorem}
\newtheorem{lemma}{Lemma}
\theoremstyle{definition}
\newtheorem{example}{Example}
\newtheorem{remark}{Remark}
\newlength{\subfigwidth}
\newlength{\subfigcolsep}
\setkeys{Gin}{width=\subfigwidth}

\usepackage{setspace}
\title{Panel Data Analysis with Heterogeneous Dynamics\footnote{First Version: February 2014.
		This paper was previously circulated under the title ``Dynamic Panel Data Analysis when the Dynamics are Heterogeneous.''}}

\author{Ryo Okui\footnote{NYU Shanghai, 1555 Century Avenue, Pudong, Shanghai, China, 200122; and Department of Economics, University of Gothenburg,
P.O. Box 640, SE-405 30 Gothenburg, Sweden.
Tel: +86-21-2059-6157. Email: \href{mailto:okui.ryo.3@gmail.com}{okui.ryo.3@gmail.com}} \and Takahide Yanagi\footnote{Graduate School of Economics, Kyoto University, Yoshida Honmachi, Sakyo, Kyoto, 606-8501, Japan.  Email: \href{mailto:yanagi@econ.kyoto-u.ac.jp}{yanagi@econ.kyoto-u.ac.jp}}}

\date{January, 2019}

\begin{document}

\maketitle

\begin{abstract}
This paper proposes a model-free approach to analyze panel data with heterogeneous dynamic structures across observational units.
We first compute the sample mean, autocovariances, and autocorrelations for each unit, and then estimate the parameters of interest based on their empirical distributions.
We then investigate the asymptotic properties of our estimators using double asymptotics and propose split-panel jackknife bias correction and inference based on the cross-sectional bootstrap.
We illustrate the usefulness of our procedures by studying the deviation dynamics of the law of one price.  Monte Carlo simulations confirm that the proposed bias correction is effective and yields valid inference in small samples.

\bigskip


\bigskip 


\end{abstract}

\newpage

\section{Introduction}

Understanding the dynamics of a potentially heterogeneous variable is an important research consideration in economics.
For instance, some researchers have investigated the price deviation of the law of one price (LOP) using panel data analyses, and a recent finding by \citet{CruciniShintaniTsuruga15} indicates that time-series persistence and volatility measures for the LOP deviation are heterogeneous across US cities and goods.
Other examples include income (e.g., \citealp{BrowningEjrnesAlvarez10}) and productivity (e.g., \citealp{HsiehKlenow09} and \citealp{GandhiNavarroRivers16}) dynamics, for which there is a considerable body of research.\footnote{For example, several researchers have pointed out the heterogeneous dynamics of the income process. These include \citet{MeghirPistaferri04} and \citet{Hospido2012}, which both suggest the importance of heterogeneity in income process volatility, and \citet{BotosaruSasaki2018}, which develops a nonparametric procedure to estimate the volatility function in the permanent--transitory model of income dynamics.}

We propose easy-to-implement procedures for analyzing the heterogeneous dynamic structure of panel data, $\{\{y_{it}\}_{t=1}^T\}_{i=1}^N$, without assuming any specific model.\footnote{The proposed procedures
	are readily available via an {\ttfamily R} package from the authors' websites.
}
To this end, we investigate the cross-sectional distributional properties of the mean, autocovariances, and autocorrelations of $y_{it}$ with heterogeneous units.
Our model-free approach is especially useful when empirical researchers are reluctant to assume specific models for heterogeneity given the threat of problems with misspecification.
Despite the voluminous literature on dynamic panel data analyses,
many studies assume specific models for the dynamics (such as autoregressive (AR) models) and homogeneity in the dynamics, allowing heterogeneity only in the mean of the process.\footnote{ 
	See \citet{Arellano03b} and \citet{Baltagi08} for excellent reviews of the existing studies on dynamic panel data analyses.
}
While several other studies also consider either heterogeneous dynamics or model-free analyses, we are unaware of any specific study that proposes panel data analysis for heterogeneous dynamics without specifying a particular model.

The distributional properties of the heterogeneous mean, autocovariances, and autocorrelations provide various pieces of information and are perhaps the most basic descriptive statistics for dynamics.
Indeed, a typical first step in analyzing time-series data is to examine these properties.
As we demonstrate in this study, the distributions of the heterogeneous mean, autocovariances, and autocorrelations can also be useful descriptive statistics for understanding heterogeneous dynamics in panel data.
For example, it would be interesting to examine the degree of heterogeneity of the LOP deviations across items such as goods and services.
In this case, the mean and variance of the heterogeneous means measure the amount and dispersion of the long-run LOP deviations across items.
We can also examine the correlation of the heterogeneous means and autocorrelations that shows whether the magnitudes of the long-run LOP deviations relate to the speed of price adjustment toward the long-run LOP deviations.
Moreover, our analysis provides the entire distribution of heterogeneity for the long-run LOP deviations or the adjustment speed.
This entire distribution could be useful to investigate, for example, whether goods and services possess different dynamics.

We derive the asymptotic properties of the empirical distributions of the estimated means, autocovariances, and autocorrelations based on double asymptotics, under which both the number of cross-sectional observations, $N$, and the length of the time series, $T$, tend to infinity.
By using empirical process theory (e.g., \citealp{vanderVaartWellner96}) and the inversion theorem for characteristic functions (e.g., \citealp{gil1951note} and \citealp{wendel1961non}), we show that the empirical distributions converge weakly to Gaussian processes under a condition for the relative magnitudes of $N$ and $T$ that is slightly ``stronger'' than $N^3 / T^4 \to 0$.
The proof is challenging because our empirical distributions are biased estimators and depend on both $N$ and $T$, so that we cannot directly apply standard empirical process techniques.
We show that the estimation error in the estimated mean, autocovariances, and autocorrelations for each unit biases the empirical distributions whose convergence rates depend on $T$.
We also derive the asymptotic distributions of the estimators for other distributional properties (e.g., the quantile function) using the functional delta method.

When we can write the parameter of interest as the expected value of a smooth function of the heterogeneous mean and/or autocovariances, we derive the exact order of the bias.
This class of parameters includes the mean, variance, and other moments (such as correlations) of the heterogeneous mean, autocovariances, and/or autocorrelations.
Importantly, we can analytically evaluate the bias, and find it has two sources.
The first is the incidental parameter problem originally discussed in \citet{NeymanScott48}.
The second arises when the smooth function is nonlinear.
We show the asymptotic distribution of the estimator under the condition $N / T ^2 \to 0$ under which both biases are negligible. 
As $T$ is often small compared with $N$ in microeconometric applications, we propose to reduce the biases using \citeauthor{DhaeneJochmans15}'s (\citeyear{DhaeneJochmans15}) split-panel jackknife.
The jackknife bias-corrected estimator is asymptotically unbiased under a weaker condition on the relative magnitudes of $N$ and $T$ and does not inflate the asymptotic variance.

We propose to use the cross-sectional bootstrapping (e.g., \citealp{GoncalvesKaffo14}) to test hypotheses and construct confidence intervals.
The bootstrap distribution is asymptotically equivalent to the distribution of the estimator, but fails to capture the bias.
We thus recommend the bootstrap based on the jackknife bias-corrected estimator because this would not suffer from large bias.

As an empirical illustration, we examine the speed of price adjustment toward the long-run LOP deviation using a panel data set of various items for different US cities.
We find statistically significant evidence that long-run LOP deviations in the item--city pairs with more persistent dynamics tend to be small and suffer from relatively large shocks.
We also find formal statistical evidence that the distribution of the LOP adjustment speed for goods differs from that for services.

We also conduct Monte Carlo simulations.
They demonstrate noticeable performances of the bootstrap inference based on the jackknife bias-corrected estimation in small samples.

\paragraph{Paper organization}
Section \ref{sec-literature} reviews the studies related to this paper, Section \ref{sec-setting} explains the setting, and Section \ref{sec-procedures} introduces the procedures.
In Section \ref{sec-asymptotics}, we derive the asymptotics of the distribution estimators, while Section \ref{sec-smooth} considers the asymptotics for estimating the expected value of a smooth function.
Section \ref{sec-extensions} presents Kolmogorov--Smirnov (KS)-type tests based on the distribution estimators and Sections \ref{sec-application} and \ref{sec-LOP-montecarlo} develop the application and simulations.
Section \ref{sec-conclusion} concludes and Appendix \ref{sec-appendix} contains the technical proofs.
The supplementary appendix provides remarks on higher-order bias correction, a test for parametric specifications, other extensions, and several mathematical proofs, applications for income and productivity, and additional Monte Carlo simulations.

\section{Related studies}\label{sec-literature}
This paper most closely relates to the literature on heterogeneous panel AR models, which capture the heterogeneity in the dynamics by allowing for unit-specific AR coefficients.
\citet{PesaranSmith95}, \citet{HsiaoPesaranTahmiscioglu99}, \citet{PesaranShinSmith99}, \citet{PhillipsMoon99}, and \citet{Pesaran06} provide such analyses.
The mean group estimator in \citet{PesaranSmith95} is identical to our estimator (without bias correction) for the mean of the heterogeneous first-order autocorrelation if their AR(1) model does not contain exogenous covariates.
\citet{HsiaoPesaranTahmiscioglu99} show that the mean group estimator is asymptotically unbiased under $N/T^2\to 0$, which is the condition we obtain without the bias correction.
Building on this literature, we aim to estimate the entire distributions of the unit-specific heterogeneous mean, autocovariances, and autocorrelations without model specifications.

Researchers have developed econometric methods for investigating features of heterogeneity other than the mean.
For example, \citet{Hospido2015} investigated the variances of individual and job effects that additively affect the income process.
As another example, \citet{BotosaruSasaki2018} estimated the volatility function in a permanent--transitory model of income dynamics. 
However, these studies assume some models and have different motivations from ours.

Elsewhere, \citet{MavroeidisSasakiWelch15} identify and estimate the distribution of the AR coefficients in heterogeneous panel AR models.
The advantage of their approach is that $T$ can be fixed.
While we impose $T \to \infty$, our method is much simpler to implement.
By contrast, the estimation method in \citet{MavroeidisSasakiWelch15} requires the maximization of a kernel-weighting function written as an integration over multiple variables.

The theoretical results for our distribution estimation relate to \cite{JochmansWeidner2018} who consider estimating the distribution of a true quantity based on a noisy measurement (e.g., an estimated quantity).
They derive the formula for the bias of their empirical distribution estimator under the assumption that their observations exhibit Gaussian errors.\footnote{An inspection of their proof leads us to surmise that their results may hold more generally as long as we assume that the distribution of the standardized quantity is homogeneous (that is, the estimated quantities satisfy a location--scale assumption).
	In our setting, heterogeneity can appear in more general ways and this generality is important because we are interested in heterogeneous dynamics.
	Indeed, we suspect that it is very difficult, if not impossible, to arrive at a setting in which all means, autocovariances, and autocorrelations exhibit heterogeneity and satisfy a location--scale assumption simultaneously.}
In contrast, the present paper does not specify parametric distributions for our observations, at the cost of not showing the exact formula for the bias of the distribution estimator.
Our results and theirs are thus complementary and thereby represent individual contributions.

In a similar motivation to us, \cite{OkuiYanagi2018} develop nonparametric kernel-smoothing estimation based on the estimated means, autocovariances, and autocorrelations for cross-sectional units in panel data.
There are several theoretical differences between the two papers, and they develop different proof techniques and obtain different results (see Remark \ref{remark:kernel} for details).
Moreover, an important practical issue also arises in the kernel estimation as the cross-sectional bootstrap is not suitable for the kernel estimation.
This issue comes from the well-known result that the bootstrap cannot capture kernel-smoothing bias, so that they propose an alternative valid inference.

Several studies propose model-free methods to investigate panel data dynamics.
For example, using long panel data, \citet{Okui10a, Okui11, Okui14} estimates autocovariances, and \citet{LeeOkuiShintani13} consider infinite-order panel AR models. 
Because we can represent a stationary time series with an infinite-order AR process under mild conditions, their approach is essentially model-free. 
However, these studies assume homogeneous dynamics.

While not directly connected, this study also relates to the recent literature on random coefficients or nonparametric panel data models with nonadditive unobserved heterogeneity.\footnote{There are also studies, such as \citet{ChernozhukovFernandezValLuo2018}, that aim to investigate the heterogeneity caused by observable covariates. However, these methods do not use the panel feature of the data and are distinct from the literature to which the present paper belongs.} 
For example, \citet{ArellanoBonhomme12} consider linear random coefficients models for panel data and discuss the identification and estimation of the distribution of random coefficients using deconvolution techniques. 
\citet{Chamberlain92} and \citet{GrahamPowell12} consider a model similar to that of \citet{ArellanoBonhomme12}, but focus on the means of random coefficients.
\citet{FernandezValLee13} examine moment restriction models with random coefficients using the generalized method of moments estimation. 
Their analysis of the smooth function of unit-specific effects closely relates to our analysis of the smooth function of means and autocovariances, at least in terms of technique. 
Finally, \citet{Evdokimov09} and \citet{Freyberger17} consider nonparametric panel regression models with unit-specific and interactive fixed effects, respectively, entering the unspecified structural function, but they do not infer heterogeneous dynamic structures.

\section{Settings}\label{sec-setting}

We observe panel data $\{\{y_{it}\}_{t=1}^{T}\}_{i=1}^{N}$ where $y_{it}$ is a scalar random variable, $i$ a cross-sectional unit, and $t$ a time period.
We assume that $\{y_{it}\}_{t=1}^T$ is independent across units.
We assume that the law of $\{y_{it}\}_{t=1}^T$ is stationary over time, but its dynamic structure may be heterogeneous.
Specifically, we consider the following data-generating process (DGP) to model the heterogeneous dynamic structure, in a spirit similar to \citet{GalvaoKato14}.
The unobserved unit-specific effect $\alpha_i$ is independently drawn from a distribution common to all units.
We then draw the time series $\{y_{it}\}_{t=1}^T $ for unit $i$ from some distribution $\mathcal{L}(\{y_{it}\}_{t=1}^T; \alpha_i)$ that may depend on $\alpha_i$.
The dynamic structure of $y_{it}$ can be heterogeneous because the realized value of $\alpha_i$ can vary across units.
For example, in an application of the LOP deviations, $\alpha_i$ might represent unobservable permanent trade costs specific to item $i$.
Note that $\alpha_i$ is an abstract parameter used to model heterogeneity in the dynamics across units and does not directly appear in the actual implementation of the procedure.
For notational simplicity, we denote ``$\cdot | \alpha_{i}$'' by ``$\cdot | i$''; that is, ``conditional on $\alpha_{i}$'' becomes ``conditional on $i$'' below.

To infer the properties of heterogeneity in a model-free manner, we aim to develop statistical tools to analyze the cross-sectional distributions of the heterogeneous means, autocovariances, and autocorrelations of $y_{it}$.
The mean for unit $i$ is $\mu_{i} \coloneqq E(y_{it}|i)$.
Note that $\mu_i$ is a random variable whose realization differs across units.
Because we assume stationarity, $\mu_{i}$ is constant over time.
Let $\gamma_{k,i} \coloneqq E((y_{it} - \mu_i) (y_{i,t-k} - \mu_i)|i)$ and $\rho_{k,i} \coloneqq \gamma_{k, i} / \gamma_{0, i}$ be the $k$-th conditional autocovariance and autocorrelation of $y_{it}$ given $\alpha_{i}$, respectively. 
Note that $\gamma_{0,i}$ is the variance for unit $i$.
To understand the possibly heterogeneous dynamics, we estimate the quantities that characterize the distributions of $\mu_i$, $\gamma_{k,i}$, and $\rho_{k,i}$.
Below, we often use the notation $\xi_i$ to represent one of $\mu_i$, $\gamma_{k,i}$, and $\rho_{k,i}$.

We consider cases in which both $N$ and $T$ are large.
For example, in our empirical illustration for the LOP deviations, we use panel data where $(N,T) = (2248, 72)$. 
A large $N$ allows us to estimate consistently the cross-sectional distributions of $\mu_i$, $\gamma_{k,i}$, and $\rho_{k,i}$. We require a large $T$ to identify and estimate $\mu_i$, $\gamma_{k,i}$, and $\rho_{k,i}$ based on the time series for each unit.

Our setting is very general and includes many situations. 
\begin{example}\label{example:AR}
The panel AR(1) model with heterogeneous coefficients, as in \citet{PesaranSmith95} and others, is a special case of our setting. 
This model is $y_{it} = c_i + \phi_i y_{i,t-1} + u_{it}$, where $c_i$ and $\phi_i$ are the unit-specific parameters, and $u_{it}$ follows a strong white noise process with variance $\sigma^2$.
In this case, $\alpha_i = (c_i, \phi_i)$, $\mu_i = c_i / (1- \phi_i)$, $\gamma_{k,i} = \sigma^2 \phi_i^k/ (1- \phi_i^2)$, and $\rho_{k,i}=\phi_i^k$.
\end{example}

\begin{example}
Our setting also includes cases in which the true DGP follows some nonlinear process. 
Suppose that $y_{it}$ is generated by $y_{it} = m(\alpha_{i}, u_{it})$, where $m$ is some function and $u_{it}$ is stationary over time and independent across units. 
In this case, $\mu_i = E( m(\alpha_{i}, \epsilon_{it}) | \alpha_i )$ and $\gamma_{k, i}$ and $\rho_{k,i}$ are the $k$-th-order autocovariance and autocorrelation of $w_{it} = y_{it} - \mu_i$ given $\alpha_i$, respectively.
\end{example}

We focus on estimating the heterogeneous mean, autocovariance, and autocorrelation structure and do not aim to recover the underlying structural form of the DGP.
We understand that addressing several important economic questions requires knowledge of the structural function of the dynamics. 
Nonetheless, the distributions of the heterogeneous means, autocovariances, and autocorrelations can be estimated relatively easily without imposing strong assumptions and can provide valuable information, even if our ultimate goal is to identify the structural form.
For example, if our procedure reveals a positive correlation between the time-series variance and the time-series persistence for the LOP deviations, the structural form on the LOP deviations should be specified so as to allow for such correlation.

\section{Procedures}\label{sec-procedures}

In this section, we present the statistical procedures to estimate the distributional characteristics of the heterogeneous mean, autocovariances, and autocorrelations.

We first estimate the mean $\mu_i$, autocovariances $\gamma_{k, i}$, and autocorrelations $\rho_{k,i}$ using the sample analogs $\hat \mu_i \coloneqq \bar y_i \coloneqq T^{-1} \sum_{t=1}^{T}y_{it}$, $\hat \gamma_{k,i} \coloneqq (T-k)^{-1} \sum_{t=k+1}^T (y_{it} - \bar y_i) (y_{i,t-k} - \bar y_{i})$, and $\hat \rho_{k,i} \coloneqq \hat \gamma_{k,i} / \hat \gamma_{0,i}$.
We write $\hat \xi_i = \hat \mu_i$, $\hat \gamma_{k,i}$, or $\hat \rho_{k,i}$ as the corresponding estimator of $\xi_i = \mu_i$, $\gamma_{k,i}$, or $\rho_{k,i}$, respectively.

We then compute the empirical distribution of $\{ \hat \xi_i\}_{i=1}^N$:
\begin{align} \label{eq-em-dis}
	\mathbb{F}_N^{\hat \xi} (a) \coloneqq \frac{1}{N} \sum_{i=1}^{N} \mathbf{1}(\hat \xi_i \leq a),
\end{align} 
where $\mathbf{1} (\cdot)$ is the indicator function and $a \in \mathbb{R}$.
This empirical cumulative distribution function (CDF) is interesting in its own right because it is an estimator of the cross-sectional CDF  of $\xi_i$, say $F_0^{\xi} (a) \coloneqq \Pr (\xi_i \le a)$.

We can estimate other distributional quantities or test some hypotheses based on the empirical distribution of $\hat \xi_i$.
For example, we can consider estimating the $\tau$-th quantile $q_{\tau}^{\xi} \coloneqq \inf \{ a \in \mathbb{R} : F_0^{\xi} (a) \ge \tau \}$ by the empirical quantile $\hat q_{\tau}^{\hat \xi} \coloneqq \inf \{ a \in \mathbb{R} : \mathbb{F}_N^{\hat \xi} (a) \ge \tau \}$.
We can also test the difference in the heterogeneous dynamic structures across distinct groups and parametric specifications for the heterogeneous means, autocovariances, or autocorrelations based on the empirical distribution.
We develop such tests based on KS-type statistics in Section \ref{sec-extensions} and the supplementary appendix.

We can also estimate a function of moments of the heterogeneous mean and autocovariances straightforwardly.
Let $S \coloneqq h(E( g(\theta_i) ) )$ be the parameter of interest, where $\theta_i$ is an $l \times 1$ vector whose elements belong to a subset of $(\mu_i, \gamma_{0,i}, \gamma_{1,i},\dots)$, and $g: \mathbb{R}^l \to \mathbb{R}^m$ and $h: \mathbb{R}^m \to \mathbb{R}^n$ are vector-valued functions.
We can estimate $S$ by the sample analog:
\begin{align}\label{eq-H}
	\hat S  \coloneqq h\left( \frac{1}{N} \sum_{i=1}^N g (\hat \theta_i )\right),
\end{align}
where $\hat \theta_i$ is the estimator corresponding to $\theta_i$.
For example, we can estimate the correlation between the mean and variance using the sample correlation of $\hat \mu_i$ and $\hat \gamma_{0, i}$.
Note that we do not need to consider a function of autocorrelations separately given the autocorrelations are functions of the autocovariances.
Section \ref{sec-smooth} investigates the asymptotics for $\hat S$ when $g$ and $h$ are \textit{smooth}, and we analytically evaluate the bias of order $O(1/T)$ in $\hat S$.

\citeauthor{DhaeneJochmans15}'s (\citeyear{DhaeneJochmans15}) split-panel jackknife can reduce the bias in $\hat S$ if $g$ and $h$ are smooth.
For example, the half-panel jackknife (HPJ) bias correction can delete the bias of order $O(1/T)$.
Suppose that $T$ is even.\footnote{If $T$ is odd, 
	we define $\bar S = (\hat S^{(1,1)} + \hat S^{(2,1)}  + \hat S^{(1,2)} + \hat S^{(2,2)} )/4$ as in \citet[page 9]{DhaeneJochmans15}, where $\hat S^{(1,1)}$, $\hat S^{(2,1)}$, $\hat S^{(1,2)}$, and $\hat S^{(2,2)}$ are the estimators of $S$ computed using $\{\{ y_{it} \}_{t=1}^{\lceil T/2 \rceil}\}_{i=1}^{N}$, $\{\{ y_{it} \}_{t=\lceil T/2 \rceil+1}^{T}\}_{i=1}^{N}$, $\{\{ y_{it} \}_{t=1}^{\lfloor T/2 \rfloor}\}_{i=1}^{N}$, and $\{\{ y_{it} \}_{t=\lfloor T/2 \rfloor +1}^{T}\}_{i=1}^{N}$, respectively. 
	Here, $\lceil \cdot \rceil$ and $\lfloor \cdot \rfloor$ are the ceiling and floor functions, respectively. 
	We note that the asymptotic properties of the HPJ estimator for odd $T$ are the same as those for even $T$. 
}
We divide the panel data into two subpanels: $\{ \{ y_{it}\}_{t=1}^{T/2}\}_{i=1}^N$ and $\{ \{ y_{it}\}_{t=T/2 +1}^{T}\}_{i=1}^N$. 
Let $\hat S^{(1)}$ and $\hat S^{(2)}$ be the estimators of $S$ computed using $\{\{ y_{it} \}_{t=1}^{T/2}\}_{i=1}^{N}$ and $\{\{ y_{it} \}_{t=T/2+1}^{T}\}_{i=1}^{N}$, respectively.
Let $\bar S \coloneqq (\hat S^{(1)} + \hat S^{(2)} )/2$.
The HPJ bias-corrected estimator of $S$ is
\begin{align} \label{eq-G-HPJ}
	\hat S^H 
	\coloneqq \hat S - (\bar S - \hat S) 
	=  2 \hat S - \bar S.
 \end{align}
The HPJ estimates the bias in $\hat S$ by $\bar S - \hat S$, and $\hat S^H$ does not show bias of order $O(1/T)$.

We may also consider a higher-order jackknife bias correction to eliminate the bias of an order higher than $O(1/T)$, as discussed in \citet{DhaeneJochmans15}.
In particular, we consider the third-order jackknife (TOJ) in the empirical application and simulations.
However, we must modify the formula in \citet{DhaeneJochmans15} for TOJ to correct higher-order biases in our setting.
See the supplementary appendix for the details of this modification.
Our Monte Carlo results below indicate that TOJ can be more successful than HPJ when higher-order biases are severe.
However, in some cases, TOJ eliminates biases at the cost of deteriorating precision of estimation.
Hence, we recommend adopting both HPJ and TOJ in practical situations.

For statistical inference of parameter $S$, we suggest using the cross-sectional bootstrap to approximate the distribution of the bias-corrected estimator. Here, we present the algorithm for the HPJ estimator.
The cross-sectional bootstrap regards each time series as the unit of observation and approximates the distribution of statistics under the empirical distribution of ($\hat \theta_i, \hat \theta_i^{(1)}, \hat \theta_i^{(2)}$), where $\hat \theta_i^{(1)}$ and $\hat \theta_i^{(2)}$ denote the estimates of $\theta_i$ from the first and second subpanels, respectively.
The algorithm is:
\begin{enumerate}
 \item Randomly draw $(\hat \theta_1^*, \hat \theta_1^{*(1)}, \hat \theta_1^{*(2)}), \dots, (\hat \theta_N^*, \hat \theta_N^{*(1)}, \hat \theta_N^{*(2)})$ from $\{ (\hat \theta_i, \hat \theta_i^{(1)}, \hat \theta_i^{(2)})\}_{i=1}^N$ with replacement.\footnote{If bootstrap is used to approximate the distribution of $\hat S-S$ not $\hat S^H-S$, then we just need to resample from $\{ \hat \theta_1,\dots, \hat \theta_N \}$. If inference is based on the TOJ estimator, estimates of $\theta_i $ from other subpanels are also required.}

 \item Compute the statistics of interest, say $\vartheta$, using $(\hat \theta_1^*, \hat \theta_1^{*(1)}, \hat \theta_1^{*(2)}), \dots, ( \hat \theta_N^*, \hat \theta_N^{*(1)}, \hat \theta_N^{*(2)})$

 \item Repeat 1 and 2 $B$ times. 
 Let $\vartheta^* (b)$ be the statistics of interest computed in the $b$-th bootstrap.

 \item Compute the quantities of interest using the empirical distribution of $\{ \vartheta^* (b) \}_{b=1}^B$.
\end{enumerate}
For example, suppose that we are interested in constructing a 95\% confidence interval for a scalar parameter $S$. 
We obtain the bootstrap approximation of the distribution of $\vartheta = \hat S^H - S$.
Let $\hat S^{H*} (b)$ be the HPJ estimate of $S$ obtained with the $b$-th bootstrap sample.
We then compute the 2.5\% and 97.5\% quantiles, denoted as $q^*_{0.025}$ and $q^*_{0.975}$, respectively, of the empirical distribution of $\{ \vartheta^* (b) \}_{b=1}^B$ with $\vartheta^* (b) = \hat S^{H*}(b) - \hat S^{H}$. 
The bootstrap 95\% confidence interval for $S$ is $[ \hat S^H - q^*_{0.975}, \hat S^H - q^*_{0.025}]$.

\section{Asymptotic analysis for the distribution estimators}\label{sec-asymptotics}

This section presents the asymptotic properties of the distribution estimator in \eqref{eq-em-dis}.
We first show the uniform consistency of the empirical distribution and then derive the functional central limit theorem (functional CLT).
We also show that the functional delta method can apply in this case.
All analyses presented below are under double asymptotics ($N, T \to \infty$).\footnote{More precisely, 
	we consider the case in which $T=T(N)$, where $T(N)$ is increasing in $N$ and $T(N) \to \infty$ as $N \to \infty$, but the exact form of the function $T(N)$ is left unspecified, except the condition imposed in each theorem.
	Note that analyses under sequential asymptotics where $N \to \infty$ after $T \to \infty$ ignore estimation errors in $\hat \mu_i$, $\hat \gamma_{k,i}$, and $\hat \rho_{k,i}$, and they fail to capture the bias. We do not consider the sequential asymptotics where $T \to \infty$ after $N \to \infty$.
}

While our setting is fully nonparametric as introduced in Section \ref{sec-setting}, the following representation is useful for our theoretical analysis.
Let $w_{it} \coloneqq y_{it} - E(y_{it}|i) = y_{it} - \mu_i$.
By construction, $y_{it} = \mu_{i} + w_{it}$ and $E(w_{it}|i)=0$ for any $i$ and $t$. 
Note also that $\gamma_{k,i} = E(w_{it} w_{i,t-k}|i)$.

\subsection{Assumptions}

Because we use empirical process techniques, it is convenient to rewrite the empirical distributions as empirical processes indexed by a class of indicator functions.
Let $\mathbb{P}_{N}^{\hat \xi} \coloneqq N^{-1} \sum_{i=1}^{N} \delta_{\hat \xi_i}$ be the empirical measure of $\hat \xi_i = \hat \mu_i$, $\hat \gamma_{k,i}$, or $\hat \rho_{k,i}$, where $\delta_{\hat \xi_i}$ is the probability distribution degenerated at $\hat \xi_i$.
Let $\mathcal{F} \coloneqq \{ \mathbf{1}_{(-\infty, a]} : a \in \mathbb{R} \}$ be the class of indicator functions where $\mathbf{1}_{(-\infty,a]} (x) \coloneqq \mathbf{1} ( x \le a )$.
We denote the probability measure of $\xi_i$ as $P_0^{\xi}$.
In this notation, the empirical distribution function $\mathbb{F}^{\hat \xi}_{N}$ in \eqref{eq-em-dis} is an empirical process indexed by $\mathcal{F}$, and $\mathbb{P}_N^{\hat \xi} f = \mathbb{F}_N^{\hat \xi} (a)$ for $f = \mathbf{1}_{(-\infty, a]}$.
Similarly, $P_0^{\xi} f = F_0^{\xi}(a) = \Pr(\xi_i \leq a)$ for $f = \mathbf{1}_{(-\infty, a]}$.
We often use shorthand notations such as $\mathbb{P}_N = \mathbb{P}_N^{\hat \xi}$, $\mathbb{F}_N = \mathbb{F}^{\hat \xi}_N$, $P_0 = P_0^{\xi}$, and $F_0 = F_0^{\xi}$ by omitting the superscripts $\hat \xi$ and $\xi$.

Throughout the study, we assume the following summarizes the conditions in Section \ref{sec-setting}.

\begin{assumption} \label{as-basic}
	The sample space of $\alpha_i$ is some Polish space and $y_{it} \in \mathbb{R}$ is a scalar real random variable. 
	$\{(\{y_{it}\}_{t=1}^T, \alpha_i)\}_{i=1}^N$ is independently and identically distributed (i.i.d.) across $i$. 
\end{assumption}

We stress that the i.i.d. assumption does not restrict the heterogeneous dynamics.

The following assumptions depend on natural numbers $r_m$ and $r_d$, which will be specified in the theorems that use this assumption. For a strictly stationary stochastic process $\{X_t\}_{t=1}^\infty$, define $\alpha$-mixing coefficients as $\alpha (m) = \sup_{A \in \mathcal{M}_1^k, B\in \mathcal{M}_{k+m}^{\infty}} |\Pr (A\cap B) - \Pr(A) \Pr(B) |$, where $\mathcal{M}_a^b$ denotes the $\sigma$-algebra generated by $X_j$ for $a \leq j \leq b$, and call the process $\alpha$-mixing if $\alpha (m) \to 0$ as $m\to \infty$.

\begin{assumption} \label{as-mixing-c}
	For each $i$, $\{y_{it}\}_{t=1}^{\infty}$ is strictly stationary and $\alpha$-mixing given $\alpha_i$, with mixing coefficients $\{\alpha (m|i)\}_{m=0}^\infty$.
	There exists a natural number $r_m$ and a sequence $\{ \alpha (m) \}_{m=0}^\infty$ such that for any $i$ and $m$, $\alpha (m|i) \le \alpha (m)$ and $\sum_{m=0}^{\infty} (m+1)^{r_m/2-1} \alpha(m) ^{\delta / (r_m+\delta)} < \infty$ for some $\delta>0$.
\end{assumption}

\begin{assumption}\label{as-w-moment-c}
	There exists a natural number $r_d$ such that $E|w_{it}|^{r_d+\delta} < \infty$ for some $\delta > 0$.
\end{assumption}

Assumptions \ref{as-mixing-c} and \ref{as-w-moment-c} are mild regularity conditions on the process of $y_{it}$.
Assumption \ref{as-mixing-c} is a mixing condition depending on $r_m$ and restricts the degree of persistence of $y_{it}$ across time.
It also imposes stationarity on $\{y_{it}\}_{t=1}^{\infty}$, which in particular implies that the initial values are generated from the stationary distribution.
Note that a large $T$ could also guarantee that such an initial value condition (see, e.g., Section 4.3.2 in \citealp{Hsiao2014}) is negligible in our analysis.
Nonetheless, we impose this condition to simplify the analysis.
Assumption \ref{as-w-moment-c} requires that $w_{it}$ has some moment higher than the $r_d$-th order.
These assumptions are satisfied, for example,
when $y_{it}$ follows a heterogeneous stationary panel ARMA model with Gaussian innovations.

We also introduce Assumptions \ref{as-mu-con}, \ref{as-gamma-con}, and \ref{as-rho-con} for the uniform consistency and functional CLTs of the empirical distributions $\mathbb{F}_N^{\hat \mu}$, $\mathbb{F}_N^{\hat \gamma_k}$, and $\mathbb{F}_N^{\hat \rho_k}$, respectively.
Condition a) in each assumption is introduced for the uniform consistency, and the remaining conditions are required for the functional CLTs.
We define $\bar w_i \coloneqq T^{-1} \sum_{t=1}^T w_{it}$.

\begin{assumption} \label{as-mu-con}
	a) The random variable $\mu_i$ is continuously distributed.
	b) The CDF of $\mu_i$ is thrice boundedly differentiable.
	c) The CDF of $\hat \mu_i$ is thrice boundedly differentiable uniformly over $T$.
	d) There exists some fixed $M < \infty$ such that $E[(\bar w_i)^2 | \mu_i = \cdot ] \le M / T$.
\end{assumption}

\begin{assumption} \label{as-gamma-con}
	a) The random variable $\gamma_{k,i}$ is continuously distributed.
	b) The CDF of $\gamma_{k,i}$ is thrice boundedly differentiable.
	c) The CDF of $\hat \gamma_{k,i}$ is thrice boundedly differentiable uniformly over $T$.
	d) There exists some fixed $M < \infty$ such that $E[(\bar w_i)^2 | \gamma_{k,i} = \cdot] \le M/T$ and $E[(\hat \gamma_{k,i} - \gamma_{k,i})^2 | \gamma_{k,i} = \cdot] \le M / T$. 
\end{assumption}

\begin{assumption} \label{as-rho-con}
	a) The random variable $\rho_{k,i}$ is continuously distributed.
	b) The CDF of $\rho_{k,i}$ is thrice boundedly differentiable.
	c) The CDF of $\hat \rho_{k,i}$ is thrice boundedly differentiable uniformly over $T$.
	d) There exists some fixed $M < \infty$ such that $E[(\bar w_i)^2 | \rho_{k,i} = \cdot] \le M/T$, $E[(\hat \gamma_{k,i} - \gamma_{k,i})^2 | \rho_{k,i} = \cdot] \le M/T$, and $E[(\hat \gamma_{0,i} - \gamma_{0,i})^2 | \rho_{k,i} = \cdot] \le M/T$.
	e) There exist some fixed $\varepsilon > 0$ and $M < \infty$ such that $\hat \gamma_{0,i} > \varepsilon$, $\gamma_{0,i} > \varepsilon$, $|\hat \gamma_{k,i}| < M$, and $|\gamma_{k,i}| < M$ almost surely.
\end{assumption}

Assumption \ref{as-mu-con} states that $\mu_i$ and $\hat \mu_i $ are continuous random variables.
This assumption is restrictive in the sense that it does not allow a discrete distribution of $\mu_i$ or no heterogeneity in the mean (i.e., $\mu_i$ is homogeneous such that $\mu_i = \mu$ for some constant $\mu$ for any $i$).\footnote{Discrete heterogeneity is considered in, for example, \cite{BonhommeManresa15} and \cite{SuShiPhillips14} for linear panel data analyses.}
\footnote{We might consider testing homogeneity in a formal manner by extending the testing procedures in \citet{pesaran2008testing} to our model-free context. 
	The construction of test statistics and the derivation of their asymptotic distributions for such extensions are nontrivial tasks, and this topic is left for future work.}
The uniform consistency and functional CLT could not hold without the continuity of $\mu_i$.
The assumption also imposes restrictions on the distribution of the noise $\bar w_{i}$. These assumptions are satisfied when data are generated by Gaussian ARMA processes with continuously distributed parameters. That said, we may be able to relax the continuity of $\hat \mu_i$ (the estimated version of $\mu_i$) in condition c), but this requires different proofs to evaluate the order of the bias for the functional CLT. 
Assumption \ref{as-mu-con} also restricts the order of the conditional moment of $\bar w_i$. Note that 
Lemma \ref{lem-moment-w} shows that the order of $E(\bar w_i^2)$ is $1/T$ and this assumption states that the same order holds for the conditional counterpart.
Assumptions \ref{as-gamma-con} and \ref{as-rho-con} are similar to Assumption \ref{as-mu-con}, except Assumption \ref{as-rho-con}.e restricts variances $\gamma_{0, i}$ and $\hat \gamma_{0,i}$ that are bounded away from zero and autocovariances $\gamma_{k,i}$ and $\hat \gamma_{k,i}$ that are bounded.
We need these additional conditions to examine the empirical distribution for $\hat \rho_{k,i}$. 

\subsection{Uniform consistency}
The following theorem establishes the uniform consistency of the distribution estimator.

\begin{theorem} \label{thm-gc}
	Suppose that Assumptions \ref{as-basic}, \ref{as-mixing-c}, \ref{as-w-moment-c}, and \ref{as-mu-con}.a hold for $r_m = 2$ and $r_d = 2$ if $\hat \xi_i = \hat \mu_i$ and $\xi_i = \mu_i$;
	that Assumptions \ref{as-basic}, \ref{as-mixing-c}, \ref{as-w-moment-c}, and \ref{as-gamma-con}.a hold for $r_m = 4$ and $r_d = 4$ if $\hat \xi_i = \hat \gamma_{k,i}$ and $\xi_i = \gamma_{k,i}$; 
	and that Assumptions \ref{as-basic}, \ref{as-mixing-c}, \ref{as-w-moment-c}, and \ref{as-rho-con}.a hold for $r_m = 4$ and $r_d = 4$ if $\hat \xi_i = \hat \rho_{k,i}$ and $\xi_i = \rho_{k,i}$.
	When $N,T \to \infty$, the class $\mathcal{F}$ is $P_0$-Glivenko--Cantelli in the sense that $\sup_{f \in \mathcal{F}} | \mathbb{P}_N f - P_{0} f | \stackrel{as}{\longrightarrow} 0$ where $\stackrel{as}{\longrightarrow}$ signifies the almost sure convergence.
\end{theorem}

Note that Theorem \ref{thm-gc} cannot be directly shown by the usual Glivenko--Cantelli theorem (e.g., Theorem 19.1 in \citealp{vanderVaart98}) because the true distribution of $\hat \xi_i$ changes as $T$ increases.
Nonetheless, our proof follows similar steps to those of the usual Glivenko--Cantelli theorem.

\subsection{Functional central limit theorem}

We present the functional CLTs for the empirical distributions of $\hat \mu_i$, $\hat \gamma_{k,i}$, and $\hat \rho_{k,i}$. 
We aim to derive the asymptotic law of $\sqrt{N}(\mathbb{P}_{N} f - P_{0} f)$ where $f \in \mathcal{F}$.
We can also obtain the asymptotic distribution of other quantities via the functional delta method based on this result.

The functional CLT for $\mathbb{P}_N$ holds under a similar set of assumptions for the uniform consistency, but we need all of the conditions in Assumption \ref{as-mu-con}, \ref{as-gamma-con}, or \ref{as-rho-con} to evaluate the order of the bias.
We also require a condition on the relative magnitudes of $N$ and $T$ asymptotically to eliminate the bias.
Let $\ell^{\infty}(\mathcal{F})$ be the collection of all bounded real functions on $\mathcal{F}$.

\begin{theorem} \label{thm-fclt}
	Suppose that Assumptions \ref{as-basic}, \ref{as-mixing-c}, \ref{as-w-moment-c}, and \ref{as-mu-con} hold for $r_m = 4$ and $r_d = 4$ if $\hat \xi_i = \hat \mu_i$ and $\xi_i = \mu_i$;
	that Assumptions \ref{as-basic}, \ref{as-mixing-c}, \ref{as-w-moment-c}, and \ref{as-gamma-con} hold for $r_m = 8$ and $r_d = 8$ if $\hat \xi_i = \hat \gamma_{k,i}$ and $\xi_i = \gamma_{k,i}$; 
	and that Assumptions \ref{as-basic}, \ref{as-mixing-c}, \ref{as-w-moment-c}, and \ref{as-rho-con} hold for $r_m = 8$ and $r_d = 8$ if $\hat \xi_i = \hat \rho_{k,i}$ and $\xi_i = \rho_{k,i}$.
	When $N,T \to \infty$ with $N^{3+\epsilon}/T^4 \to 0$ for some $\epsilon \in (0, 1/3)$, we have 
	\begin{align*}
		\sqrt{N}(\mathbb{P}_N - P_0) \leadsto \mathbb{G}_{P_0} \qquad \mbox{in} \quad \ell^{\infty}(\mathcal{F}),
	\end{align*}
	where $\leadsto$ signifies weak convergence, $\mathbb{G}_{P_0}$ is a Gaussian process with zero mean and covariance function $E(\mathbb{G}_{P_0}(f_{i})\mathbb{G}_{P_0}(f_j))=F_0(a_i \wedge a_j) - F_0(a_i) F_0(a_j)$ with $f_i = \mathbf{1}_{(-\infty, a_i]}$ and $f_{j} = \mathbf{1}_{(-\infty, a_j]}$ for $a_i, a_j \in \mathbb{R}$ and $a_i \wedge a_j$ is the minimum of $a_i$ and $a_j$.
\end{theorem}

The asymptotic law of the empirical process is Gaussian, which is identical to the limiting distribution for the empirical process constructed using the true $\xi_i = \mu_i$, $\gamma_{k,i}$, or $\rho_{k,i}$. 
However, this result requires that $N^{3+\epsilon}/T^4 \to 0$ for some $\epsilon$ such that $0 < \epsilon < 1/3$, which allows us to ignore the estimation error in $\hat \xi_i = \hat \mu_i$, $\hat \gamma_{k,i}$, or $\hat \rho_{k,i}$ asymptotically.
Note that the condition, $N^{3+\epsilon}/T^4 \to 0$, is almost equivalent to $N^3 / T^4 \to 0$ because we can select an arbitrarily small $\epsilon > 0$.

We provide a brief summary of the proof and explain why we require the condition $N^{3+\epsilon}/T^4 \to 0$. 
The key to understanding the mechanism behind the requirement that $N^{3+\epsilon}/T^4 \to 0$ is to recognize that $E (\mathbb{P}_N f) \neq P_0 f$. 
That is, $\mathbb{P}_N f$ is not an unbiased estimator for $P_0 f$.
As a result, we cannot directly apply the existing results for the empirical process to derive the asymptotic distribution.
Let $P_T = P_T^{\hat \xi}$ be the (true) probability measure of $\hat \xi_i = \hat \mu_i$, $\hat \gamma_{k,i}$, or $\hat \rho_{k,i}$.
Note that $P_T$ depends on $T$ and $P_T \neq P_0$, and observe that $E(\mathbb{P}_N f) = P_T f$.
Let $\mathbb{G}_{N, P_T} \coloneqq \sqrt{N}(\mathbb{P}_N - P_T)$.
We observe that
\begin{align}
	\sqrt{N}(\mathbb{P}_N f - P_0 f) 
	= & \ \mathbb{G}_{N,P_T} f \label{fclt_1} \\
	&  + \sqrt{N} (P_T f - P_0 f). \label{fclt_2}
\end{align}
For $\mathbb{G}_{N,P_T}$ in \eqref{fclt_1}, we can directly apply the uniform CLT for the empirical process based on triangular arrays (\citealp[Lemma 2.8.7]{vanderVaartWellner96}) and obtain $\mathbb{G}_{N, P_T}  \leadsto \mathbb{G}_{P_0}$ in $\ell^{\infty}(\mathcal{F})$ as $N \to \infty$.
This part of the proof is standard.

We require the condition $N^{3+\epsilon}/T^4 \to 0$ to eliminate the effect of the bias term $\sqrt{N}(P_T f - P_0 f )$ in \eqref{fclt_2}.
In the proof of the theorem, we show that
\begin{align*}
	\sqrt{N} \sup_{f\in\mathcal{F}} \left| P_{T} f - P_{0}f \right| = O\left(\frac{\sqrt{N}}{T^{2/(3+\epsilon)}}\right),
\end{align*}
for any $0 < \epsilon < 1/3$.
The result is based on the evaluation that the difference between $P_T$ and $P_0$ is of order $O(1 / T^{2/(3+\epsilon)})$. This order is obtained by evaluating the characteristic functions of $\hat \xi_i$ and $\xi_i$, and applying the inversion theorem (\citealp{gil1951note} and \citealp{wendel1961non}).
We note that the condition, $0 < \epsilon < 1/3$, is used to ensure the integrability of integrals for the inversion theorem (see the proof for details).
As a result, we can establish the weak convergence under the condition $N^{3+\epsilon}/T^4 \to 0$ for a sufficiently small $0 < \epsilon < 1/3$.

\begin{remark}
	When we additionally assume that $\hat \xi_i$ exhibits a Gaussian error, we can derive the exact bias that is of order $O(1/T)$ and show the same limiting law as in Theorem \ref{thm-fclt} under the weaker condition on the relative magnitudes that $N / T^2 \to 0$.
	In this case, we can also validate the HPJ bias correction for the distribution estimator.
	The proof utilizing the Gaussian assumption (and a location--scale assumption) can be found in \citet{JochmansWeidner2018} in a general setting for noisy measurement, and we do not explore such a proof here.
	However, we stress that our proof for Theorem \ref{thm-fclt} is distinct from theirs because we do not assume Gaussianity nor any parametric specification for $\hat \xi_i$ and $\xi_i$ and it requires a quite different proof technique.
\end{remark}

\subsection{Functional delta method}

We can derive the asymptotic distribution of an estimator that is a function of the empirical distribution using the functional delta method.
Suppose that we are interested in the asymptotics of $\phi(\mathbb{P}_{N})$ for a functional $\phi:D(\mathcal{F}) \to \mathbb{R}$ where $D(\mathcal{F})$ is the collection of all c\`adl\`ag real functions of $\mathcal{F}$.
For example, the $\tau$-th quantile $\phi (P_0) = q_{\tau} = F_0^{-1}(\tau) = \inf\{a \in \mathbb{R} : F_0 (a) \geq \tau \}$ for $\tau \in (0,1)$ may be estimated by the empirical quantile of $\hat \xi_i$: $\phi(\mathbb{P}_{N}) = \hat q_{\tau}= \mathbb{F}_N^{-1} (\tau) = \inf\{a \in \mathbb{R} : \mathbb{F}_N (a) \geq \tau \}$.
More generally, we can estimate the quantile process $F_0^{-1}$ using the empirical quantile process $\mathbb{F}_N^{-1}$.

The derivation of the asymptotic distribution of $\phi(\mathbb{P}_{N})$ is a direct application of the functional delta method (e.g., \citealp[Theorems 3.9.4]{vanderVaartWellner96}) and Theorem \ref{thm-fclt}.
We summarize this result in the following corollary.\footnote{In Corollary \ref{cor-delta}, we can change the Hadamard differentiability of $\phi:D(\mathcal{F}) \subset \ell^{\infty}(\mathcal{F}) \to \mathbb{E}$ to the Hadamard differentiability of $\phi: \ell^{\infty}(\mathcal{F}) \subset D(\mathbb{\bar R}) \to \mathbb{R}$ tangentially to a set of continuous functions in $D(\mathbb{\bar R})$, where $D(\mathbb{\bar R})$ is the Banach space of all c\`adl\`ag functions $z:\mathbb{\bar {R}} \to \mathbb{R}$ on $\mathbb{\bar R}$ equipped with the uniform norm.
	See Lemma 3.9.20 and Example 3.9.21 in \citet{vanderVaartWellner96} for details.
}

\begin{corollary} \label{cor-delta}
	Suppose that the assumptions in Theorem \ref{thm-fclt} hold.	
	Suppose that $\phi:D(\mathcal{F}) \subset \ell^{\infty}(\mathcal{F}) \to \mathbb{E}$ is Hadamard differentiable at $P_0$ with the derivative $\phi'_{P_0}$ where $\mathbb{E}$ is a normed linear space.
	When $N,T \to \infty$ with $N^{3+\epsilon}/T^4 \to 0$ for some $\epsilon \in (0, 1/3)$, we have $\sqrt{N} ( \phi(\mathbb{P}_N) - \phi(P_{0}) ) \leadsto \phi'_{P_{0}}(\mathbb{G}_{P_{0}})$.
\end{corollary}

As an example, we can use this result to derive the asymptotic distribution of $\hat q_{\tau}$.
The form $\phi'_{P_0}$ for $\hat q_{\tau}$ is available in Example 20.5 in \citet{vanderVaart98} and indicates that $\sqrt{N} (\hat q_{\tau} - q_{\tau}) \leadsto \mathcal{N} ( 0, \tau (1-\tau) / (f (q_\tau))^2)$ where $\mathcal{N}(\mu, \sigma^2)$ is the normal distribution with mean $\mu$ and variance $\sigma^2$ and $f = f^{\xi}$ is the density function of $\xi_i$.
We can also derive the asymptotic law of the empirical quantile process $\mathbb{F}_N^{-1}$.
If $f$ is continuous and positive in the interval $[F_0^{-1}(p) - \varepsilon, F_0^{-1}(q) + \varepsilon]$ for some $0 < p < q < 1$ and $\varepsilon > 0$, then Corollary \ref{cor-delta} means that
\begin{align*}
	\sqrt{N} \left( \mathbb{F}_N^{-1} - F_0^{-1} \right)
	\leadsto - \frac{\mathbb{G}_{P_0} \circ F_0 \left( F_0^{-1} \right)}{f \left( F_0^{-1} \right)}
	\qquad \mbox{in} \quad \ell^{\infty}[p, q].
\end{align*}
This process is known to be Gaussian with zero mean and a known covariance function (e.g., Example 3.9.24 in \citealp{vanderVaartWellner96}).

\section{Function of the expected value of a smooth function of the heterogeneous mean and/or autocovariances}\label{sec-smooth}

In this section, we consider the estimation of a function of the expected value of a smooth function of the heterogeneous mean and/or autocovariances.
We also develop the asymptotic justifications of the HPJ bias correction and the cross-sectional bootstrap inference.

\subsection{Asymptotic results}\label{sec-ghat}

We derive the asymptotic properties of $\hat S = h(N^{-1} \sum_{i=1}^{N} g( \hat \theta_i ))$ in \eqref{eq-H} as the estimator of $S = h(E(g(\theta_i)))$.
Define $G \coloneqq E(g(\theta_i))$ and $\hat G \coloneqq N^{-1} \sum_{i=1}^N g(\hat \theta_i)$ such that $S = h(G)$ and $\hat S = h(\hat G)$.

We make the following assumptions to develop the asymptotic properties of $\hat S$.

\begin{assumption}\label{as-h-c}
The function $h: \mathbb{R}^m \to \mathbb{R}^n$ is continuous in a neighborhood of $G$.
\end{assumption}

\begin{assumption}\label{as-h-cd}
The function $h: \mathbb{R}^m \to \mathbb{R}^n$ is continuously differentiable in a neighborhood of $G$.
The matrix of the first derivatives $\nabla h(G)\coloneqq (\nabla h_1(G)^\top, \nabla h_2(G)^\top, \dots, \nabla h_n(G)^\top)^\top$ is of full row rank.
\end{assumption}

\begin{assumption}\label{as-multi-smooth}
The function $g=(g_1,g_2,\dots,g_m):\mathcal{O} \to \mathbb{R}^m$ is twice-continuously differentiable where $\mathcal{O} \subset \mathbb{R}^{l}$ is a convex open subset.
The covariance matrix $\Gamma \coloneqq E[(g(\theta_i) - E(g(\theta_i)))(g(\theta_i) - E(g(\theta_i)))^\top]$ exists and is nonsingular.
For any $p=1,2,\dots,m$, the elements of the Hessian matrix of $g_p$ are bounded functions.
For any $p=1,2,\dots,m$, the function $g_p$ satisfies $E[((\partial/ \partial z_{j}) g_p(z)|_{z=\theta_i})^4] < \infty$ for any $j=1,2,\dots,m$.
\end{assumption}

These assumptions impose conditions on the smoothness of $h$ and $g$ and the existence of moments.
Assumption \ref{as-h-c} applies to the continuous mapping theorem for the proof of consistency.
Assumption \ref{as-h-cd} is stronger than Assumption \ref{as-h-c} and is used for the application of the delta method to derive the asymptotic distribution.
Assumption \ref{as-multi-smooth} states that the function $g$ is sufficiently smooth.
This assumption is satisfied when the parameter of interest is the mean (i.e., $g(a) = a$) or the $p$-th order moment (i.e., $g(a) =a ^p$), for example.
However, this assumption is not satisfied when estimating the CDF (i.e., $g(a) = \mathbf{1} (a \leq c)$ for some $c\in\mathbb{R}$) or quantiles.
The existence of the first derivative is crucial for analyzing the asymptotic property of $\hat S$.
The second derivative is useful for evaluating the order of the asymptotic bias.
Assumption \ref{as-multi-smooth} also guarantees that the asymptotic variance exists, which rules out homogeneous dynamics, i.e., it excludes the case where $\theta_i = \theta$ for constant $\theta$ for any $i$ (see the supplementary appendix for the asymptotic results for homogeneous dynamics).

The following theorem demonstrates the asymptotic properties of $\hat S$.

\begin{theorem}\label{thm-h}
Let $r^* = 4$ if $\theta_i = \mu_i$ such that $S =h( E(g(\mu_i)))$ for some $h$ and $g$, and $r^*=8$ if $\theta_i$ contains $\gamma_{k,i}$ for some $k$. 
Suppose that Assumptions \ref{as-basic}, \ref{as-mixing-c}, \ref{as-w-moment-c}, \ref{as-h-c}, and \ref{as-multi-smooth} hold for $r_m=4$ and $r_d=r^*$.
When $N,T \to \infty$, it holds that $\hat S  \stackrel{p}{\longrightarrow}  S$.
Moreover, suppose that Assumption \ref{as-h-cd} also holds.
When $N, T \to \infty$ with $N/T^2 \to 0$, it holds that
\begin{align*}
 \sqrt{N} (\hat S - S) \leadsto
\mathcal{N} \left(0,  \nabla h(G) \Gamma (\nabla h(G))^\top \right).
\end{align*}
\end{theorem}

The estimator $\hat S$ is consistent when both $N$ and $T$ tend to infinity and is asymptotically normal with mean zero when $N/T^2 \to 0$.
Importantly, in contrast to the discussion in Section \ref{sec-asymptotics}, the distribution of $\theta_i$ need not be continuous and can be discrete as long as it is not degenerate (homogeneous).
The remarkable result is that the asymptotically unbiased estimation holds under $N/T^2 \to 0$.
This condition is weaker than that for $\mathbb{P}_N$, which is $N^{3 + \epsilon}/ T^{4} \to 0$ for some $0 < \epsilon < 1/3$.
This result comes from the smoothness of $g$ and the fact that $\hat \theta_i$ is first-order unbiased for $\theta_i$.
\citet{FernandezValLee13} also observe similar asymptotic results for estimators of smooth functions of heterogeneous quantities in a different context.

To obtain a better understanding of the results in the theorem, we first consider the case in which $\theta_i = \mu_i$ such that $l=1$, $h$ is an identity function such that $n=1$, and $g$ is a scalar function such that $m=1$. 
Denote our parameter of interest as $G^{\mu} \coloneqq E(g(\mu_i))$ and let $\hat G^{\hat \mu} \coloneqq N^{-1}\sum_{i=1}^N g(\hat \mu_i)$.
By Taylor's theorem and $\hat \mu_i = \mu_i + \bar w_i$, we observe the following expansion:
\begin{align}
	\sqrt{N}\left( \hat G^{\hat \mu} - G^{\mu} \right)
	= \frac{1}{\sqrt{N}} \sum_{i=1}^N \Big( g(\mu_i) - E \big( g(\mu_i) \big) \Big)
	+ \frac{1}{\sqrt{N}} \sum_{i=1}^N \bar w_{i} g'(\mu_i)
	+ \frac{1}{2\sqrt{N}} \sum_{i=1}^N (\bar w_i)^2 g''(\tilde{\mu}_i),\label{eq-mu-expansion}
\end{align}
where $\tilde{\mu}_i$ is between $\mu_i$ and $\hat \mu_i$. 
The second term in \eqref{eq-mu-expansion} has a mean of zero and is of order $O_{p}(1/\sqrt{T})$.
The fact that it has a mean of zero is the key reason that a milder condition, $N/T^2 \to 0$, is sufficient for the asymptotically unbiased estimation of $G^{\mu}$.
The third term corresponds to the bias caused by the nonlinearity of $g$. 
When $g$ is linear, this term does not appear and the parameter can be estimated without any restriction on the relative magnitudes of $N$ and $T$. 
The nonlinearity bias is of order $O_{p}(\sqrt{N}/T)$.
We use the condition $N/T^2 \to 0$ to eliminate the effect of this bias.

When our parameter of interest involves $\gamma_{k,i}$ for some $k$, we encounter an additional source of bias. 
Let us consider the case in which $\theta_i = \gamma_{k,i}$ for some $k$ such that $l=1$, $h$ is an identity function such that $n=1$, and $g$ is a scalar function such that $m=1$. 
We denote our parameter of interest as $G^{\gamma_{k}} \coloneqq E(g(\gamma_{k,i}))$ and let $\hat G^{\hat \gamma_k} \coloneqq N^{-1}\sum_{i=1}^N g(\hat \gamma_{k,i})$.
We can expand $\hat \gamma_{k,i}$ as follows:
\begin{align*}
	\hat \gamma_{k,i} = \gamma_{k, i} + \frac{1}{T-k} \sum_{t=k+1}^T (w_{it} w_{i,t-k} - \gamma_{k,i})  - ( \bar w_i)^2 + o_{p}\left(\frac{1}{T}\right).
\end{align*}
Note that the second term has a mean of zero, although it is of order $O_p(1/\sqrt{T})$.
The third term $(\bar w_i)^2$ is the estimation error in $\bar y_i$ ($=\hat \mu_i$) and is of order $O_p(1/T)$, and causes the incidental parameter bias \citep{NeymanScott48,Nickell1981}.
By Taylor's theorem and the expansion of $\hat \gamma_{k, i}$, we have
\begin{align}
	& \sqrt{N}(\hat G^{\hat \gamma_k} - G^{\gamma_k} ) \nonumber \\
	=& \ \frac{1}{\sqrt{N}} \sum_{i=1}^N \Big( g(\gamma_{k,i}) - E \big( g(\gamma_{k,i}) \big) \Big) \label{gex-1}\\
	& \ + \frac{1}{\sqrt{N}} \sum_{i=1}^N \left(  \frac{1}{T-k} \sum_{t=k+1}^T w_{it} w_{i,t-k} - \gamma_{k,i}\right) g'(\gamma_{k,i}) \label{gex-2}\\
	& \ - \frac{1}{\sqrt{N}} \sum_{i=1}^N (\bar w_{i})^2 g'(\gamma_{k,i}) +
	\frac{1}{2\sqrt{N}} \sum_{i=1}^N  (\hat \gamma_{k,i} - \gamma_{k,i})^2 g^{\prime\prime}(\tilde \gamma_{k,i})
	+ o_{p}\left( \frac{\sqrt{N}}{T}\right) \label{gex-3}, 
\end{align}
where $\tilde \gamma_{k,i}$ is between $\hat \gamma_{k,i}$ and $\gamma_{k,i}$.
In contrast to $\hat G^{\hat \mu}$, this $\hat G^{\hat \gamma_k}$ has an incidental parameter bias corresponding to the first term in \eqref{gex-3}.
This bias is of order $O_{p}(\sqrt{N}/T)$ and does not appear in the expansion of $\hat G^{\hat \mu}$. 
This term makes the condition $N/T^2\to 0$ necessary, even when $g$ is linear.
The other terms are similar to those in the expansion of $\hat G^{\hat \mu}$.
The term on the right-hand side of \eqref{gex-1} yields the asymptotic normality of $\hat G^{\hat \gamma_k}$. 
The term in \eqref{gex-2} has a mean of zero and is of order $O_{p}(1/\sqrt{T})$.
The second term in \eqref{gex-3} is the nonlinearity bias term that also appears in $\hat G^{\hat \mu}$, which is also of order $O_{p}(\sqrt{N}/T)$.

\begin{remark}\label{remark:kernel}
	The analysis here is not applicable to kernel-smoothing estimation, and asymptotic analyses for kernel-smoothing estimators require different proof techniques.
	To see this, we consider the kernel estimator for the density of $\mu_i$, say $(Nh)^{-1} \sum_{i=1}^N K((x - \hat \mu_i) / h)$, where $K$ is a kernel function and $h \to 0$ is bandwidth.
	The summand $K((x - \cdot) / h)$ depends on the bandwidth $h$, which shrinks to zero as the sample size increases, so that the shape of the summand changes depending on the sample size, unlike the summand $g(\cdot)$ here.
	As a result, the kernel estimation requires much more careful investigations for nonlinearity bias terms.
	\citet{OkuiYanagi2018} formally demonstrate this issue for the kernel density and CDF estimation and find that their relative magnitude conditions of $N$ and $T$ differ from the condition $N / T^2 \to 0$ here and vary in the number of nonlinearity bias terms that can be evaluated.
\end{remark}

\subsection{Split-panel jackknife bias correction}

We provide a theoretical justification for the HPJ bias-corrected estimator in \eqref{eq-G-HPJ}, which we base on the bias-correction method proposed by \citet{DhaeneJochmans15}.
We make the following additional assumptions to study the HPJ estimator of $S$.

\begin{assumption} \label{as-g-moment-2}
	The function $g=(g_1,g_2,\dots,g_m):\mathcal{O} \to \mathbb{R}^m$ is thrice differentiable.
	The covariance matrix $\Gamma = E[(g(\theta_i) - E(g(\theta_i)))(g(\theta_i) - E(g(\theta_i)))^\top]$ exists and is nonsingular.
	For any $p=1,2,\dots,m$, the function $g_p$ satisfies $E[((\partial/ \partial z_{j}) g_p(z)|_{z=\theta_i})^4] < \infty$ for any $j=1,2,\dots,m$, and $E[((\partial^2 / \partial z_{j_1}  \partial z_{j_2})g_p(z)|_{z=\theta_i})^4]<\infty$ for any $j_1, j_2 =1,2,\dots,m$.
	All third-order derivatives of $g$ are bounded.
\end{assumption}

Assumption \ref{as-g-moment-2} requires that $g$ is thrice differentiable, contrary to Assumption \ref{as-multi-smooth} and imposes stronger moment conditions.
We require this condition to conduct a higher-order expansion of $\hat S$.

The following theorem shows the asymptotic normality of the HPJ estimator.

\begin{theorem} \label{thm-hpj}
	Let $r^* = 8$ if $\theta_i = \mu_i$ such that $S = h( E(g(\mu_i)))$ for some $h$ and $g$, and $r^*=16$ if $\theta_i$ contains $\gamma_{k,i}$ for some $k$. 
	Suppose that Assumptions \ref{as-basic}, \ref{as-mixing-c}, \ref{as-w-moment-c}, \ref{as-h-cd}, and \ref{as-g-moment-2} are satisfied for $r_m = 8$ and $r_d =r^*$. 
	When $N,T\to \infty$ with $N/T^2 \to \nu$ for some $\nu \in [0,\infty)$, it holds that
	\begin{align*}
	\sqrt{N} (\hat S^H - S) \leadsto \mathcal{N} \left(0,  \nabla h(G) \Gamma (\nabla h(G))^\top \right).
	\end{align*}
\end{theorem}

The HPJ estimator is asymptotically unbiased, even when $N/T^2 \to 0$ is not satisfied.
Moreover, this bias correction does not inflate the asymptotic variance.
The reason why the HPJ works is the same as \citet{DhaeneJochmans15}, and the detail can be found in the proof of Theorem \ref{thm-hpj}.

\subsection{Cross-sectional bootstrap}

In this section, we present the justification for the use of the cross-sectional bootstrap introduced in Section \ref{sec-procedures}.
The first theorem concerns $\hat S$ and the second theorem discusses the case where $\hat S^H$.
We also provide a theorem for distribution function estimators.

We require several additional assumptions. 
The following assumption is required for Lyapunov's conditions for $\hat G^*$, which is the estimator of $G$ obtained with the bootstrap sample.
Note that $\hat S^*= h(\hat G^*) = h(N^{-1} \sum_{i=1}^N g(\hat \theta_i^*))$ where $\hat \theta_i^*$ is the estimator of $\theta_i$ based on the bootstrap sample.

\begin{assumption}\label{as-boot}
	The function $g=(g_1,g_2,\dots,g_m):\mathcal{O} \to \mathbb{R}^m$ is twice-continuously differentiable.
	The covariance matrix of $g(\theta_i)$, $\Gamma$, exists and is nonsingular. 
	The elements of the Hessian matrices of $g_p$ for $p=1,2,\dots, m$, $g_{p_1} (\cdot) g_{p_2} (\cdot)$ for $p_1, p_2 = 1,2,\dots, m$, and $(g(\cdot) ^\top g(\cdot))$ are bounded.
	For any $p=1,2,\dots,m$, the function $g_p$ satisfies $E[( (\partial/\partial z_{j})g_p(z)|_{z=\theta_i})^4]<\infty$ for any $j=1,2,\dots,l$.
	For any $p_1, p_2 = 1,2,\dots, m$, $E[ ( (\partial / \partial z_j) g_{p_1}(z) |_{z=\theta_i} g_{p_2} (\theta_i))^2] < \infty$.
	For any $j=1,2,\dots, l$, $E[ ( g(\theta_i)^\top g(\theta_i) (\partial / \partial z_j) g_{p_1}(z) |_{z=\theta_i} g_{p_2} (\theta_i))^2] < \infty$ is satisfied.
\end{assumption}

The following theorem states that the bootstrap distribution converges to the asymptotic distribution of $\hat S$, but fails to capture the bias term.
Let $P^*$ be the bootstrap distribution (that is identical here to the empirical distribution of $\hat \theta_i$, or as below, $\hat \xi_i$).

\begin{theorem} \label{thm-bootstrap}
	Let $r^* = 4$ if $\theta_i = \mu_i$ such that $S =h( E(g(\mu_i)))$ for some $h$ and $g$, and $r^*=8$ if $\theta_i$ contains $\gamma_{k,i}$ for some $k$.
	Suppose that Assumptions \ref{as-basic}, \ref{as-mixing-c}, \ref{as-w-moment-c}, \ref{as-h-cd}, and \ref{as-boot} hold for $r_m = 4$ and $r_d =r^*$.
	When $N, T \to \infty$, we have 
	\begin{align*}
		\sup_{x \in \mathbb{R}} \left| P^* \left(  \sqrt{N} (\hat S^* - \hat S)  \le x \right) - \Pr \left( \mathcal{N} \left(0,  \nabla h(G) \Gamma (\nabla h(G))^\top \right) \le x \right) \right| \stackrel{p}{\longrightarrow} 0.
	\end{align*}
\end{theorem}

The bootstrap does not capture the bias properties of $\hat G$ shown in Section \ref{sec-ghat}.
This implies that when $T$ is small, we must be cautious about using the bootstrap to make statistical inference.
\citet{GalvaoKato14}, \citet{GoncalvesKaffo14}, and \citet{Kaffo14} also observe similar issues.

We can also show that the bootstrap can approximate the asymptotic distribution of the HPJ estimator.
The proof is analogous to the proof of Theorem \ref{thm-bootstrap}, and is thus omitted.
\begin{theorem}
	Let $r^* = 4$ if $\theta_i = \mu_i$ such that $S =h( E(g(\mu_i)))$ for some $h$ and $g$, and $r^* = 8$ if $\theta_i$ contains $\gamma_{k,i}$ for some $k$.
	Suppose that Assumptions \ref{as-basic}, \ref{as-mixing-c}, \ref{as-w-moment-c}, \ref{as-h-cd}, and \ref{as-boot} are satisfied for $r_m = 4$ and $r_d =r^*$.
	When $N,T \to \infty$, we have 
	\begin{align*}
		\sup_{x \in \mathbb{R}} \left| P^* \left(  \sqrt{N} (\hat S^{H*} - \hat S^H) \le x \right) - \Pr \left( \mathcal{N} \left(0,  \nabla h(G) \Gamma (\nabla h(G))^\top \right) \le x \right) \right| \stackrel{p}{\longrightarrow} 0.
	\end{align*}
\end{theorem}

The cross-sectional bootstrap can approximate the asymptotic distribution of the HPJ estimator correctly under the condition that $N/T^2 $ does not diverge.
Because the HPJ estimator has a smaller bias,
the bootstrap approximation is more appropriate for the HPJ estimator.

Lastly, we show the pointwise validity of the bootstrap for the estimator of the distribution function evaluated at some point $a \in \mathbb{R}$.\footnote{While a uniform validity of the cross-sectional bootstrap for the distribution estimator would be desirable, 
	this investigation is challenging because it requires new empirical process techniques. 
	For now, we leave it as an interesting future research topic.
}
Let $\xi_i$ be one of $\mu_i$, $\gamma_{k,i}$, and $\rho_{k,i}$ with the distribution function $F_0 = F_0^\xi$, and $\hat \xi_i$ be the corresponding estimator.
The pointwise estimator of $F_0(a)$ at $a \in \mathbb{R}$ is $\mathbb{F}_N (a) = N^{-1} \sum_{i=1}^N \mathbf{1}(\hat \xi_i \le a)$.
The bootstrap estimator is $\mathbb{F}_N^* (a) \coloneqq N^{-1} \sum_{i=1}^N \mathbf{1}(\hat \xi_i^* \le a)$.

\begin{theorem} \label{thm-dist-bootstrap}
	Suppose that the assumptions in Theorem \ref{thm-gc} hold.
	When $N,T \to \infty$, it holds that
	\begin{align*}
		\sup_{x \in \mathbb{R}} \left| P^* \left(\sqrt{N}\left(\mathbb{F}_N^{*} (a) - \mathbb{F}_N(a)\right) \le x \right) - \Pr\Big( \mathcal{N} \Big(0, F(a)(1 - F(a)) \Big) \le x \Big) \right| \stackrel{p}{\longrightarrow} 0.
	\end{align*}
\end{theorem}

If the rate condition $N^{3 + \epsilon} / T^4 \to 0$ is satisfied for some $\epsilon \in (0, 1/3)$, the bootstrap distribution consistently estimates the asymptotic distribution of $\mathbb{F}_N(a)$.

\section{Difference in degrees of heterogeneity}\label{sec-extensions}

We develop a two-sample KS test as an application of the convergence for the distribution estimator in Section \ref{sec-asymptotics}.
Specifically, we develop a test to examine whether the distributions of $\mu_i$, $\gamma_{k,i}$, and $\rho_{k,i}$ differ across distinct groups.
In many applications, it would be interesting to see whether distinct groups possess different heterogeneous structures.
For example, when studying the LOP deviation, we may want to know whether the distribution of LOP adjustment speed differs between goods and services.
We develop a testing procedure for such hypotheses without any parametric specification.
We consider two panel data sets for two different groups: $\{\{y_{it,(1)} \}_{t=1}^{T_{1}} \}_{i=1}^{N_1}$ and $\{\{y_{it,(2)}\}_{t=1}^{T_{2}}\}_{i=1}^{N_2}$.
We allow $T_1 \neq T_2$ and/or $N_1 \neq N_2$.
Define $y_{i,(1)} \coloneqq \{y_{it,(1)}\}_{t=1}^{T_{1}}$ and  $y_{i,(2)} \coloneqq \{y_{it,(2)}\}_{t=1}^{T_{2}}$.

We estimate the distributions of the mean, autocovariances, or autocorrelations for each group.
Let $\xi_{i,(a)} = \mu_{i,(a)}$, $\gamma_{k,i,(a)}$, or $\rho_{k,i,(a)}$ be the true quantity for group $a=1,2$.
Let $\hat \xi_{i,(a)} = \hat \mu_{i,(a)}$, $\hat \gamma_{k,i,(a)}$, or $\hat \rho_{k,i,(a)}$ be the corresponding estimator of $\xi_{i,(a)}$.
We denote the probability distribution of $\xi_{i,(a)}$ by $P_{0,(a)} = P_{0,(a)}^{\xi}$ and the empirical distribution of $\hat \xi_{i,(a)}$ by $\mathbb{P}_{N_a,(a)} = \mathbb{P}_{N_a,(a)}^{\hat \xi}$ for $a=1,2$. 

We focus on the following hypothesis to examine the difference in the degrees of heterogeneity between the two groups.
\begin{align*}
		H_{0} : P_{0,(1)} = P_{0,(2)} \;  \mbox{ v.s. } \;  H_{1} : P_{0,(1)} \neq P_{0,(2)}.
\end{align*}
Under the null hypothesis $H_0$, the distributions are identical for the two groups. 

We investigate the hypothesis using the following two-sample KS statistic based on our empirical distribution estimators.
\begin{align*}
	KS_2 & \coloneqq \sqrt{\frac{N_{1}N_{2}}{N_1 + N_2}} \left\| \mathbb{P}_{N_1,(1)} - \mathbb{P}_{N_2,(2)} \right\|_{\infty}
	= \sqrt{\frac{N_{1}N_{2}}{N_1 + N_2}} \sup_{f \in \mathcal{F}} \left| \mathbb{P}_{N_1,(1)} f - \mathbb{P}_{N_2,(2)} f \right|,
\end{align*}
where $\| \cdot \|_{\infty}$ is the uniform norm.
This measures the distance between the empirical distributions of the two groups and differs from the usual two-sample KS statistic in that it is based on the empirical distributions of the estimates.

We introduce the following assumption about the data sets.

\begin{assumption} \label{as-ks}
	Each of $\{\{y_{it,(1)} \}_{t=1}^{T_{1}} \}_{i=1}^{N_1}$ and $\{\{y_{it,(2)}\}_{t=1}^{T_{2}}\}_{i=1}^{N_2}$ satisfies Assumptions \ref{as-basic}, \ref{as-mixing-c}, \ref{as-w-moment-c}, and \ref{as-mu-con} with $r_m = 4$, $r_d = 4$ when $\hat \xi_{i,(a)} = \hat \mu_{i,(a)}$ and $\xi_{i,(a)} = \mu_{i,(a)}$; 
	Assumptions \ref{as-basic}, \ref{as-mixing-c}, \ref{as-w-moment-c}, and \ref{as-gamma-con} with $r_m = 8$, $r_d = 8$ when $\hat \xi_{i,(a)} = \hat \gamma_{k,i,(a)}$ and $\xi_{i,(a)} = \gamma_{k,i,(a)}$;
	and Assumptions \ref{as-basic}, \ref{as-mixing-c}, \ref{as-w-moment-c}, and \ref{as-rho-con} with $r_m = 8$, $r_d = 8$ when $\hat \xi_{i,(a)} = \hat \rho_{k,i,(a)}$ and $\xi_{i,(a)} = \rho_{k,i,(a)}$.
	$(y_{1,(1)}, \dots, y_{N_1,(1)})$ and $(y_{1,(2)}, \dots, y_{N_2, (2)})$ are independent.
\end{assumption}

We need the assumptions introduced in the previous sections along with the independence assumption, which implies that our test cannot be used to determine the equivalence of the distributions of two variables from the same units.
Our test is intended to compare the distributions of the same variable from different groups.
It is also important to note that this independence assumption may collapse when there are some time effects.
For example, when the time periods of the two panel data sets overlap, the panel data sets can be dependent given the presence of common time trends.

The asymptotic null distribution of $KS_2$ is derived using Theorem \ref{thm-fclt}.

\begin{theorem} \label{thm-KS2}
	Suppose that Assumption \ref{as-ks} is satisfied.
	When $N_{1},T_{1} \to \infty$ with $N_{1}^{3+\epsilon}/T_{1}^4 \to 0$ and $N_{2},T_{2} \to \infty$ with $N_{2}^{3+\epsilon}/T_{2}^4 \to 0$ for some $\epsilon \in (0, 1/3)$ and $N_1 / (N_1 + N_2) \to \lambda$ for some $\lambda \in (0,1)$, it holds that $KS_2$ converges in a distribution to $\| \mathbb{G}_{P_{0,(1)}} \|_{\infty}$ under $H_{0}$.
\end{theorem}

The asymptotic null distribution of $KS_2$ is the uniform norm of a Gaussian process.
We require the conditions $N_1^{3 + \epsilon} / T_1^{4} \to 0$ and $N_2^{3 + \epsilon} / T_2^{4} \to 0$  to use the result of Theorem \ref{thm-fclt}.
The condition $N_1/(N_1+N_2) \to \lambda$ implies that $N_1$ is not much greater or less than $N_2$ and guarantees the existence of the asymptotic null distribution.

Note that the asymptotic distribution does not depend on $P_{0,(1)}$,  and  critical values can be computed readily.
\citet{Kolmogorov1933} and \citet{Smirnov1944} (for easy reference see, e.g., Theorem 6.10 in \citealp{Shao03} or Section 2.1.5 in \citealp{Serfling02}) showed that
\begin{align} \label{KS-01}
\Pr(\| \mathbb{G}_{P_{0,(1)}} \|_{\infty} \leq a)  = 1- 2\sum_{j=1}^{\infty} (-1)^{j-1} \exp\left(-2 j^2 a^2 \right),
\end{align}
for any continuous distribution $P_{0,(1)}$, with $a > 0$.
The right-hand side of \eqref{KS-01} does not depend on $P_{0,(1)}$.
Moreover, the critical values are readily available in many statistical software packages and it is easy to implement our tests.

\begin{remark}
	When the true distributions of the estimated quantities, $\hat \xi_{i,(1)}$ and $\hat \xi_{i,(2)}$, are the same, i.e., when $P_{T_1,(1)}^{\hat \xi} = P_{T_2,(2)}^{\hat \xi}$, neither the condition $N_1^{3 + \epsilon} / T_1^4 \to 0$ nor $N_2^{3 + \epsilon} / T_2^4 \to 0$ is needed to establish Theorem \ref{thm-KS2}.  
	In particular, when $T_1 = T_2$ and the mean and dynamic structures of the two groups are completely identical under the null hypothesis, we can test the null hypothesis $H_0$ without restricting the relative order of $N_a$ and $T_a$ for $a=1,2$. In this case, both of the distribution function estimates suffer from the same bias, which is canceled out in $KS_2$ under the null hypothesis.
	Note that we still need the condition $N_1/(N_1 + N_2) \to \lambda \in (0,1)$.
\end{remark}

\section{Empirical application}\label{sec-application}

We apply our procedures to panel data on prices in US cities.
The speed of price adjustment toward the long-run law of one price (LOP) has important implications in economics.
\cite{AndersonVanWincoop04} survey the literature on price adjustment and trade costs.
Several studies focus on the properties of heterogeneity in price deviations from the LOP based on model specifications.
For example, \cite{EngelRogers01} and \cite{ParsleyWei01} examine the heterogeneity in the time-series volatility of the LOP deviation, and \cite{CruciniShintaniTsuruga15} consider the heterogeneous properties for the time-series persistence of the LOP deviation.

We investigate the heterogeneous properties of the LOP deviations across cities and items using our procedures.
We examine whether the LOP deviations dynamics are heterogeneous depending on the item-specific unobserved component such as city- or item-specific permanent trade costs or the item category (e.g., goods or services).
Our model-free empirical results complement the findings in existing studies by investigating the heterogeneous properties of the permanent amount, time-series volatility, and persistence of the LOP deviations across cities and items.

\subsection{Data}

We use data from the American Chamber of Commerce Researchers Association Cost of Living Index produced by the Council of Community and Economic Research.\footnote{Mototsugu Shintani kindly provided us with the dataset ready for analysis.}
\cite{ParsleyWei96}, \cite{YazganYilmazkuday11}, and 
\cite{CruciniShintaniTsuruga15} use the same data set as  
\cite{LeeOkuiShintani13} do for their empirical illustration.
The data set contains quarterly price series of 48 consumer price index categorized goods and services for 52 US cities from 1990Q1 to 2007Q4.\footnote{While the original data source contains price information for more items in more cities, we restrict the observations to obtain a balanced panel data set, as in \cite{CruciniShintaniTsuruga15}.} 

The LOP deviation for item $k$ in city $i$ at time $t$ is $y_{i k t}=\ln p_{i k t}-\ln p_{0kt}$ where $p_{ikt}$ is the price of item $k$ in city $i$ at time $t$ and $p_{0kt}$ is that for the benchmark city of Albuquerque, NM.
In this empirical application, we regard each item--city pair as a cross-sectional unit, implying a focus on the heterogeneity of the dynamic structures of the LOP deviations across item--city pairs. 
The number of units is $N=2448$ ($= 48 \times 51$) and the length of the time series is $T=72$ ($=18 \times 4$).

\subsection{Results}

Table \ref{table-accra} summarizes the estimates based on the empirical distribution without bias correction (ED) and the HPJ and TOJ estimates for the distributional features of the heterogeneous means, variances, and first-order autocorrelations of the LOP deviations. 
The estimates of mean, standard deviation (std), 25\% quantile (Q25), median (Q50), and 75\% quantile (Q75) for each quantity are presented. 
We also estimate the correlations between these three quantities.
The 95\% confidence intervals are computed using the cross-sectional bootstrap.

\begin{table}[!t]
	\begin{footnotesize}
		\caption{Distribution of price dynamics}
		\label{table-accra} 
		\begin{center}
			\subcaption{Distributions of $\mu$, $\gamma_0$, and $\rho_1$}
			\begin{tabular}{lccccc}
				\hline\hline
				\multicolumn{1}{l}{}&\multicolumn{1}{c}{mean}&\multicolumn{1}{c}{std}&\multicolumn{1}{c}{Q25}&\multicolumn{1}{c}{Q50}&\multicolumn{1}{c}{Q75}\tabularnewline
				\hline
				{\bfseries Distribution of $\mu$}&&&&&\tabularnewline
				~~ED&-0.037&0.131&-0.121&-0.040&0.033\tabularnewline
				~~95\% CI&[-0.042, -0.032]&[0.127, 0.136]&[-0.129, -0.115]&[-0.046, -0.035]&[0.027, 0.041]\tabularnewline
				~~HPJ&-0.037&0.120&-0.115&-0.040&0.024\tabularnewline
				~~95\% CI&[-0.042, -0.032]&[0.115, 0.125]&[-0.125, -0.108]&[-0.050, -0.033]&[0.018, 0.034]\tabularnewline
				~~TOJ&-0.037&0.114&-0.110&-0.040&0.014\tabularnewline
				~~95\% CI&[-0.042, -0.032]&[0.108, 0.119]&[-0.127, -0.097]&[-0.055, -0.030]&[0.003, 0.034]\tabularnewline
				\hline
				{\bfseries Distribution of $\gamma_0$}&&&&&\tabularnewline
				~~ED&0.021&0.021&0.010&0.016&0.025\tabularnewline
				~~95\% CI&[0.020, 0.022]&[0.018, 0.024]&[0.009, 0.010]&[0.015, 0.017]&[0.024, 0.026]\tabularnewline
				~~HPJ&0.024&0.019&0.013&0.019&0.029\tabularnewline
				~~95\% CI&[0.023, 0.025]&[0.017, 0.021]&[0.012, 0.013]&[0.018, 0.020]&[0.028, 0.031]\tabularnewline
				~~TOJ&0.026&0.017&0.015&0.021&0.032\tabularnewline
				~~95\% CI&[0.024, 0.027]&[0.014, 0.020]&[0.014, 0.016]&[0.020, 0.023]&[0.029, 0.034]\tabularnewline
				\hline
				{\bfseries Distribution of $\rho_1$}&&&&&\tabularnewline
				~~ED&0.531&0.206&0.386&0.543&0.691\tabularnewline
				~~95\% CI&[0.523, 0.539]&[0.201, 0.211]&[0.372, 0.393]&[0.533, 0.555]&[0.679, 0.702]\tabularnewline
				~~HPJ&0.627&0.159&0.519&0.644&0.748\tabularnewline
				~~95\% CI&[0.616, 0.636]&[0.150, 0.168]&[0.497, 0.532]&[0.628, 0.664]&[0.729, 0.767]\tabularnewline
				~~TOJ&0.663&0.134&0.58&0.686&0.760\tabularnewline
				~~95\% CI&[0.645, 0.679]&[0.119, 0.149]&[0.533, 0.612]&[0.651, 0.723]&[0.723, 0.796]\tabularnewline
				\hline
			\end{tabular}
			\medskip
			\subcaption{Correlation structure}
			\begin{tabular}{lccc}
				\hline\hline
				\multicolumn{1}{l}{}&\multicolumn{1}{c}{$\mu$ vs $\gamma_0$}&\multicolumn{1}{c}{$\mu$ vs $\rho_1$}&\multicolumn{1}{c}{$\gamma_0$ vs $\rho_1$}\tabularnewline
				\hline
				ED&0.046&-0.128&0.113\tabularnewline
				95\% CI&[0.002, 0.092]&[-0.166, -0.087]&[0.048, 0.177]\tabularnewline
				HPJ&0.082&-0.160&0.106\tabularnewline
				95\% CI&[0.009, 0.143]&[-0.218, -0.104]&[-0.003, 0.205]\tabularnewline
				TOJ&0.123&-0.170&0.020\tabularnewline
				95\% CI&[0.025, 0.233]&[-0.260, -0.086]&[-0.151, 0.167]\tabularnewline
				\hline
			\end{tabular}
		\end{center}
	\end{footnotesize}
\end{table}

The results show that the bias-corrected estimates can substantially differ from the ED estimates, even though this data set has a relatively long time series.
In particular, both HPJ and TOJ estimates imply more volatile and persistent dynamics than those implied by the ED estimates.
This result demonstrates that the bias correction is important even when $T$ is relatively large.

LOP deviations exhibit significant heterogeneity across item--city pairs, as shown in the estimates of the standard deviations and quantiles of the heterogeneous means, variances, and first-order autocorrelations. 
The standard deviation estimate of the heterogeneous mean indicates a substantial degree of permanent price differences across cities and items.
Likewise, the magnitude of the variance in price differences shows large heterogeneity.
Interestingly, the positive correlation between the means and the variances implies that the larger the permanent LOP deviation is, the larger the variance of the deviation tends to be.

The results for the first-order autocorrelations indicate that the LOP deviations are serially positively correlated.
The amount of heterogeneity implied by the bias-corrected estimates is less than that implied by the ED estimate, but all estimates imply that the first-order autocorrelations have a substantial degree of heterogeneity.
The first-order autocorrelations are negatively correlated with the mean LOP deviation, and the correlation between first-order autocorrelations and variances is slightly positive.
This result indicates that item--city pairs with persistent price difference tend to have small permanent price differences but tend to suffer from relatively large shocks.

We also examine whether the distribution of the LOP deviations dynamics differs between goods and services, similarly to prior works that point out different price dynamics between goods and services (e.g., \citealp{ParsleyWei96} and \citealp{NakamuraSteinsson08}).
Table \ref{table-items} describes the classification of goods and services in our analysis.
Table \ref{table-accra-gs} summarizes the estimation results, which indicate that price dynamics for goods are markedly different from those of services.
In particular, prices for services tend to have more persistent dynamics. 
In fact, the value of the two-sample KS test for the first-order autocorrelations is 0.353, with a $p$-value of 0.
This provides statistical evidence that the speed of price adjustment for services is slower than that for goods.

\begin{table}[!t]
	\centering
	\caption{Items}
	\label{table-items}
	\begin{tabular}{|l|l|}
		\hline 
		Goods  & T-bone steak, Ground beef, Frying chicken, Chunk light tuna, Whole milk, \\
		& Eggs, Margarine, Parmesan cheese, Potatoes, Bananas, Lettuce, Bread, Coffee,  \\
		&Sugar, Corn flakes, Sweet peas, Peaches, Shortening, Frozen corn, Soft drink, \\
		& Beer, Wine, Facial tissues, Dishwashing powder, Men's dress shirt, Shampoo, \\
		& Toothpaste, Tennis balls.
		\\
		\hline
		Services & Hamburger sandwich, Pizza, Fried chicken, Total home energy cost, Telephone, \\
		& Apartment, Home purchase price, Mortgage rate, Monthly payment,  \\
		&Dry cleaning, Major appliance repair, Auto maintenance, Gasoline, \\
		&Doctor office visit, Dentist office visit, Haircut, Beauty salon, \\
		& Newspaper subscription, Movie, Bowling.
		\\
		\hline
	\end{tabular}
	\medskip
	\begin{minipage}{\linewidth}
		Note: The service category includes those that may be considered as goods, but whose prices are likely to include the cost of a service.
		Our results are robust to minor modifications to the classification.
	\end{minipage}
\end{table}

\begin{table}[!t]
	\begin{footnotesize}
		\caption{Distribution of price dynamics for goods and services}
		\label{table-accra-gs} 
		\begin{center}
			\subcaption{Distributions of $\mu$, $\gamma_0$, and $\rho_1$ for goods}
			\begin{tabular}{lccccc}
				\hline\hline
				\multicolumn{1}{l}{}&\multicolumn{1}{c}{mean}&\multicolumn{1}{c}{std}&\multicolumn{1}{c}{Q25}&\multicolumn{1}{c}{Q50}&\multicolumn{1}{c}{Q75}\tabularnewline
				\hline
				{\bfseries Distribution of $\mu$}&&&&&\tabularnewline
				~~ED&-0.042&0.120&-0.127&-0.051&0.026\tabularnewline
				~~95\% CI&[-0.048, -0.036]&[0.114, 0.125]&[-0.134, -0.119]&[-0.059, -0.045]&[0.016, 0.032]\tabularnewline
				~~HPJ&-0.042&0.107&-0.121&-0.053&0.019\tabularnewline
				~~95\% CI&[-0.049, -0.036]&[0.100, 0.113]&[-0.131, -0.109]&[-0.064, -0.046]&[0.006, 0.027]\tabularnewline
				~~TOJ&-0.042&0.101&-0.115&-0.054&0.016\tabularnewline
				~~95\% CI&[-0.048, -0.036]&[0.094, 0.107]&[-0.132, -0.094]&[-0.073, -0.040]&[-0.005, 0.035]\tabularnewline
				\hline
				{\bfseries Distribution of $\gamma_0$}&&&&&\tabularnewline
				~~ED&0.025&0.021&0.013&0.019&0.029\tabularnewline
				~~95\% CI&[0.024, 0.026]&[0.019, 0.024]&[0.012, 0.014]&[0.019, 0.020]&[0.028, 0.031]\tabularnewline
				~~HPJ&0.028&0.022&0.016&0.023&0.033\tabularnewline
				~~95\% CI&[0.027, 0.030]&[0.019, 0.025]&[0.015, 0.017]&[0.022, 0.024]&[0.031, 0.035]\tabularnewline
				~~TOJ&0.030&0.021&0.018&0.025&0.033\tabularnewline
				~~95\% CI&[0.029, 0.032]&[0.018, 0.025]&[0.017, 0.020]&[0.023, 0.027]&[0.030, 0.038]\tabularnewline
				\hline
				{\bfseries Distribution of $\rho_1$}&&&&&\tabularnewline
				~~ED&0.474&0.188&0.348&0.480&0.604\tabularnewline
				~~95\% CI&[0.465, 0.484]&[0.181, 0.194]&[0.333, 0.360]&[0.466, 0.491]&[0.591, 0.617]\tabularnewline
				~~HPJ&0.577&0.140&0.491&0.591&0.673\tabularnewline
				~~95\% CI&[0.564, 0.591]&[0.130, 0.150]&[0.463, 0.509]&[0.564, 0.608]&[0.654, 0.697]\tabularnewline
				~~TOJ&0.604&0.106&0.553&0.630&0.686\tabularnewline
				~~95\% CI&[0.583, 0.625]&[0.086, 0.125]&[0.502, 0.594]&[0.581, 0.668]&[0.649, 0.740]\tabularnewline
				\hline
			\end{tabular}
			\medskip
			\subcaption{Distributions of $\mu$, $\gamma_0$, and $\rho_1$ for services}
			\begin{tabular}{lccccc}
				\hline\hline
				\multicolumn{1}{l}{}&\multicolumn{1}{c}{mean}&\multicolumn{1}{c}{std}&\multicolumn{1}{c}{Q25}&\multicolumn{1}{c}{Q50}&\multicolumn{1}{c}{Q75}\tabularnewline
				\hline
				{\bfseries Distribution of $\mu$}&&&&&\tabularnewline
				~~ED&-0.030&0.146&-0.113&-0.023&0.046\tabularnewline
				~~95\% CI&[-0.039, -0.022]&[0.138, 0.154]&[-0.126, -0.101]&[-0.031, -0.015]&[0.034, 0.057]\tabularnewline
				~~HPJ&-0.030&0.138&-0.104&-0.024&0.038\tabularnewline
				~~95\% CI&[-0.039, -0.021]&[0.129, 0.146]&[-0.118, -0.091]&[-0.035, -0.013]&[0.024, 0.055]\tabularnewline
				~~TOJ&-0.030&0.131&-0.093&-0.029&0.028\tabularnewline
				~~95\% CI&[-0.039, -0.021]&[0.123, 0.140]&[-0.117, -0.068]&[-0.046, -0.011]&[0.005, 0.055]\tabularnewline
				\hline
				{\bfseries Distribution of $\gamma_0$}&&&&&\tabularnewline
				~~ED&0.015&0.019&0.006&0.011&0.018\tabularnewline
				~~95\% CI&[0.014, 0.016]&[0.011, 0.025]&[0.006, 0.007]&[0.010, 0.012]&[0.018, 0.020]\tabularnewline
				~~HPJ&0.018&0.016&0.008&0.014&0.021\tabularnewline
				~~95\% CI&[0.016, 0.019]&[0.011, 0.022]&[0.007, 0.009]&[0.013, 0.016]&[0.020, 0.025]\tabularnewline
				~~TOJ&0.019&0.018&0.009&0.016&0.023\tabularnewline
				~~95\% CI&[0.018, 0.021]&[0.009, 0.027]&[0.008, 0.010]&[0.014, 0.018]&[0.020, 0.028]\tabularnewline
				\hline
				{\bfseries Distribution of $\rho_1$}&&&&&\tabularnewline
				~~ED&0.610&0.204&0.479&0.655&0.764\tabularnewline
				~~95\% CI&[0.598, 0.624]&[0.194, 0.213]&[0.453, 0.510]&[0.643, 0.669]&[0.754, 0.778]\tabularnewline
				~~HPJ&0.696&0.161&0.604&0.749&0.809\tabularnewline
				~~95\% CI&[0.679, 0.713]&[0.145, 0.174]&[0.565, 0.655]&[0.731, 0.769]&[0.790, 0.828]\tabularnewline
				~~TOJ&0.745&0.130&0.684&0.806&0.822\tabularnewline
				~~95\% CI&[0.721, 0.769]&[0.103, 0.155]&[0.613, 0.778]&[0.766, 0.851]&[0.781, 0.867]\tabularnewline
				\hline
			\end{tabular}
		\end{center}
	\end{footnotesize}
\end{table}

Our findings are informative in their own right and are in line with existing results.
For example, \citet{CruciniShintaniTsuruga15} find significant heterogeneity in LOP deviation dynamics by considering city--city pairs in addition to the city--item pairs.
\citet{choi2007heterogeneity} find that the speed of price adjustment is heterogeneous, even among tradable goods using Japanese data. 
As existing studies comparing goods and services, \citet{ParsleyWei96} and \citet{NakamuraSteinsson08} find that services exhibit slower price adjustments and less frequent price changes.
Those findings are based on model specifications for heterogeneity, so that our results complement them in a model-free manner with formal statistical procedures.

There could be several potential sources of significant heterogeneity in LOP deviation dynamics.
For example, there may be some item- and/or city-specific unobservables, such as permanent trade costs and productivity shocks, which can be sources of heterogeneity across items and cities.
As another example, an information difference across items and/or cities may lead to heterogeneous dynamic structures. 
From this viewpoint, \citet{CruciniShintaniTsuruga15} relate heterogeneity to an information difference across managers in different cities based on the noisy information model.

The empirical findings here are useful even when the ultimate goal of an application on the LOP deviations is a structural estimation based on some model specifications.
For example, our findings here illustrate the importance of taking into account heterogeneity and the type of heterogeneity that needs to be considered in structural estimation.
In particular, we demonstrate that goods and services exhibit different heterogeneous dynamic structures, so that empirical researchers should consider different heterogeneity for goods and services.
Our recommendation is thus to implement the model-free procedure for understanding the properties of heterogeneous dynamics even when investigating the underlying mechanism and their implications based on structural estimation.

\section{Monte Carlo simulation} \label{sec-LOP-montecarlo}
This section presents the Monte Carlo simulation results. 
We conduct the simulation using {\ttfamily R} with 5,000 replications.

\subsection{Design}
For $N = 250, 1000, 4000$ and $T = 12, 24, 48$, we generate simulated data using an AR(1) process
\begin{align*}
	y_{it}= (1-\phi_i)\varsigma_i + \phi_i y_{i,t-1} + \sqrt{(1-\phi_i^2)\sigma_i^2} u_{it},
\end{align*}
where $u_{it} \sim i.i.d. \ \mathcal{N}(0,1)$.
The initial observations are generated by $y_{i0} \sim i.i.d. \ \mathcal{N}(\varsigma_i, \sigma_i^2)$ and $u_{i0} \sim i.i.d. \ \mathcal{N}(0, 1)$.
Note that this DGP satisfies $\mu_i = \varsigma_i$, $\gamma_{0,i}=\sigma^2_i$, and $\rho_{1,i}=\phi_i$.
The unit-specific random variables $\varsigma_i, \phi_i$, and $\sigma_i^2$ are generated by the truncated normal distribution:
\begin{equation*}
\begin{pmatrix}
\varsigma_i \\
\sigma_i^2 \\
\phi_i
\end{pmatrix}
\sim i.i.d. \; \mathcal{N}
\begin{pmatrix}
\begin{pmatrix}
-1 \\
1.5 \\
0.4
\end{pmatrix}
,&
\begin{pmatrix}
1 & 0.2 \cdot 1 \cdot 0.7 &  -0.3 \cdot 1 \cdot 0.2 \\
0.2 \cdot 1 \cdot 0.7  & 0.7^2 & 0.4 \cdot 0.7 \cdot 0.2 \\
-0.3 \cdot 1 \cdot 0.2  & 0.4 \cdot 0.7 \cdot 0.2 & 0.2^2
\end{pmatrix}
\end{pmatrix}
,
\end{equation*}
conditional on $\sigma_i^2 > 0$ and $|\phi_i| < 1$.

\paragraph{Parameters.}
We estimate the means, standard deviations, 25\%, 50\%, and 75\% quantiles, and correlation coefficients of $\mu_i, \gamma_{0,i}$, and $\rho_{1,i}$.

\paragraph{Estimators.}
We consider three estimators:  the empirical distribution (ED) without bias correction, the HPJ bias-corrected estimator, and the TOJ bias-corrected estimator. 

\subsection{Results}
Tables \ref{table-monte-mean}, \ref{table-monte-acov}, \ref{table-monte-acor}, and \ref{table-monte-cor} summarize the results of the Monte Carlo simulation and provide the bias and the root mean squared error (rmse) of each estimator and the coverage probability (cp) of the 95\% confidence interval based on the cross-sectional bootstrap. 
The column labeled ``true'' displays the true value of the corresponding quantity.

\begin{table}[!t]
	\begin{footnotesize}
		\caption{Monte Carlo simulation results for $\mu$}
		\label{table-monte-mean} 
		\begin{center}
			\begin{tabular}{lrrrcrrrcrrrcrrr}
				\hline\hline
				\multicolumn{1}{l}{\bfseries }&\multicolumn{3}{c}{\bfseries }&\multicolumn{1}{c}{\bfseries }&\multicolumn{3}{c}{\bfseries ED}&\multicolumn{1}{c}{\bfseries }&\multicolumn{3}{c}{\bfseries HPJ}&\multicolumn{1}{c}{\bfseries }&\multicolumn{3}{c}{\bfseries TOJ}\tabularnewline
				\cline{6-8} \cline{10-12} \cline{14-16}
				\multicolumn{1}{l}{}&\multicolumn{1}{c}{true}&\multicolumn{1}{c}{$N$}&\multicolumn{1}{c}{$T$}&\multicolumn{1}{c}{}&\multicolumn{1}{c}{bias}&\multicolumn{1}{c}{rmse}&\multicolumn{1}{c}{cp}&\multicolumn{1}{c}{}&\multicolumn{1}{c}{bias}&\multicolumn{1}{c}{rmse}&\multicolumn{1}{c}{cp}&\multicolumn{1}{c}{}&\multicolumn{1}{c}{bias}&\multicolumn{1}{c}{rmse}&\multicolumn{1}{c}{cp}\tabularnewline
				\hline
				$\mu$ mean&$-0.993$&$ 250$&$12$&&$ 0.000$&$0.071$&$0.948$&&$ 0.000$&$0.071$&$0.947$&&$ 0.000$&$0.071$&$0.950$\tabularnewline
				&$-0.993$&$ 250$&$24$&&$ 0.001$&$0.067$&$0.945$&&$ 0.001$&$0.067$&$0.947$&&$ 0.001$&$0.067$&$0.947$\tabularnewline
				&$-0.993$&$ 250$&$48$&&$ 0.000$&$0.064$&$0.955$&&$ 0.000$&$0.064$&$0.955$&&$ 0.000$&$0.064$&$0.956$\tabularnewline
				&$-0.993$&$1000$&$12$&&$ 0.000$&$0.035$&$0.954$&&$ 0.000$&$0.035$&$0.952$&&$ 0.000$&$0.035$&$0.954$\tabularnewline
				&$-0.993$&$1000$&$24$&&$-0.001$&$0.034$&$0.945$&&$-0.001$&$0.034$&$0.947$&&$-0.001$&$0.034$&$0.946$\tabularnewline
				&$-0.993$&$1000$&$48$&&$-0.001$&$0.033$&$0.946$&&$-0.001$&$0.033$&$0.947$&&$-0.001$&$0.033$&$0.946$\tabularnewline
				&$-0.993$&$4000$&$12$&&$ 0.000$&$0.018$&$0.947$&&$ 0.000$&$0.018$&$0.949$&&$ 0.000$&$0.018$&$0.945$\tabularnewline
				&$-0.993$&$4000$&$24$&&$ 0.000$&$0.017$&$0.948$&&$ 0.000$&$0.017$&$0.951$&&$ 0.000$&$0.017$&$0.948$\tabularnewline
				&$-0.993$&$4000$&$48$&&$ 0.000$&$0.017$&$0.944$&&$ 0.000$&$0.017$&$0.943$&&$ 0.000$&$0.017$&$0.940$\tabularnewline
				\hline
				$\mu$ std&$ 0.997$&$ 250$&$12$&&$ 0.135$&$0.145$&$0.248$&&$ 0.042$&$0.071$&$0.901$&&$ 0.005$&$0.064$&$0.946$\tabularnewline
				&$ 0.997$&$ 250$&$24$&&$ 0.076$&$0.091$&$0.676$&&$ 0.016$&$0.055$&$0.939$&&$ 0.000$&$0.057$&$0.934$\tabularnewline
				&$ 0.997$&$ 250$&$48$&&$ 0.040$&$0.061$&$0.891$&&$ 0.004$&$0.049$&$0.943$&&$-0.001$&$0.050$&$0.939$\tabularnewline
				&$ 0.997$&$1000$&$12$&&$ 0.136$&$0.139$&$0.000$&&$ 0.044$&$0.052$&$0.673$&&$ 0.007$&$0.033$&$0.942$\tabularnewline
				&$ 0.997$&$1000$&$24$&&$ 0.076$&$0.080$&$0.121$&&$ 0.015$&$0.030$&$0.922$&&$-0.001$&$0.028$&$0.946$\tabularnewline
				&$ 0.997$&$1000$&$48$&&$ 0.040$&$0.046$&$0.613$&&$ 0.004$&$0.025$&$0.949$&&$-0.001$&$0.025$&$0.944$\tabularnewline
				&$ 0.997$&$4000$&$12$&&$ 0.137$&$0.137$&$0.000$&&$ 0.044$&$0.047$&$0.110$&&$ 0.007$&$0.017$&$0.931$\tabularnewline
				&$ 0.997$&$4000$&$24$&&$ 0.076$&$0.077$&$0.000$&&$ 0.016$&$0.020$&$0.780$&&$-0.001$&$0.014$&$0.945$\tabularnewline
				&$ 0.997$&$4000$&$48$&&$ 0.041$&$0.042$&$0.067$&&$ 0.005$&$0.013$&$0.934$&&$-0.001$&$0.013$&$0.946$\tabularnewline
				\hline
				$\mu$ 25\%Q&$-1.666$&$ 250$&$12$&&$-0.095$&$0.133$&$0.831$&&$-0.030$&$0.125$&$0.960$&&$-0.001$&$0.200$&$0.989$\tabularnewline
				&$-1.666$&$ 250$&$24$&&$-0.054$&$0.103$&$0.910$&&$-0.010$&$0.108$&$0.964$&&$ 0.003$&$0.169$&$0.992$\tabularnewline
				&$-1.666$&$ 250$&$48$&&$-0.026$&$0.091$&$0.941$&&$ 0.001$&$0.103$&$0.966$&&$ 0.006$&$0.152$&$0.993$\tabularnewline
				&$-1.666$&$1000$&$12$&&$-0.099$&$0.109$&$0.431$&&$-0.034$&$0.069$&$0.928$&&$-0.005$&$0.099$&$0.976$\tabularnewline
				&$-1.666$&$1000$&$24$&&$-0.056$&$0.072$&$0.754$&&$-0.012$&$0.056$&$0.960$&&$ 0.001$&$0.088$&$0.981$\tabularnewline
				&$-1.666$&$1000$&$48$&&$-0.029$&$0.053$&$0.896$&&$-0.003$&$0.052$&$0.958$&&$ 0.002$&$0.077$&$0.982$\tabularnewline
				&$-1.666$&$4000$&$12$&&$-0.101$&$0.103$&$0.010$&&$-0.036$&$0.048$&$0.778$&&$-0.007$&$0.051$&$0.958$\tabularnewline
				&$-1.666$&$4000$&$24$&&$-0.056$&$0.060$&$0.291$&&$-0.012$&$0.030$&$0.937$&&$ 0.002$&$0.044$&$0.970$\tabularnewline
				&$-1.666$&$4000$&$48$&&$-0.030$&$0.037$&$0.717$&&$-0.003$&$0.027$&$0.948$&&$ 0.002$&$0.039$&$0.966$\tabularnewline
				\hline
				$\mu$ 50\%Q&$-0.993$&$ 250$&$12$&&$-0.027$&$0.092$&$0.934$&&$-0.018$&$0.113$&$0.958$&&$-0.006$&$0.180$&$0.988$\tabularnewline
				&$-0.993$&$ 250$&$24$&&$-0.019$&$0.085$&$0.943$&&$-0.010$&$0.102$&$0.965$&&$-0.006$&$0.157$&$0.991$\tabularnewline
				&$-0.993$&$ 250$&$48$&&$-0.011$&$0.082$&$0.949$&&$-0.004$&$0.097$&$0.961$&&$ 0.000$&$0.143$&$0.990$\tabularnewline
				&$-0.993$&$1000$&$12$&&$-0.028$&$0.052$&$0.904$&&$-0.019$&$0.059$&$0.942$&&$-0.010$&$0.092$&$0.972$\tabularnewline
				&$-0.993$&$1000$&$24$&&$-0.019$&$0.046$&$0.926$&&$-0.009$&$0.053$&$0.952$&&$-0.001$&$0.082$&$0.978$\tabularnewline
				&$-0.993$&$1000$&$48$&&$-0.012$&$0.042$&$0.940$&&$-0.005$&$0.048$&$0.961$&&$-0.002$&$0.072$&$0.979$\tabularnewline
				&$-0.993$&$4000$&$12$&&$-0.029$&$0.036$&$0.746$&&$-0.020$&$0.035$&$0.887$&&$-0.011$&$0.048$&$0.951$\tabularnewline
				&$-0.993$&$4000$&$24$&&$-0.018$&$0.028$&$0.860$&&$-0.009$&$0.027$&$0.944$&&$-0.003$&$0.041$&$0.964$\tabularnewline
				&$-0.993$&$4000$&$48$&&$-0.011$&$0.024$&$0.904$&&$-0.004$&$0.025$&$0.950$&&$-0.001$&$0.036$&$0.967$\tabularnewline
				\hline
				$\mu$ 75\%Q&$-0.320$&$ 250$&$12$&&$ 0.064$&$0.118$&$0.909$&&$ 0.008$&$0.129$&$0.965$&&$-0.010$&$0.214$&$0.988$\tabularnewline
				&$-0.320$&$ 250$&$24$&&$ 0.035$&$0.099$&$0.940$&&$ 0.003$&$0.116$&$0.971$&&$-0.002$&$0.183$&$0.993$\tabularnewline
				&$-0.320$&$ 250$&$48$&&$ 0.014$&$0.089$&$0.955$&&$-0.003$&$0.106$&$0.969$&&$-0.005$&$0.161$&$0.994$\tabularnewline
				&$-0.320$&$1000$&$12$&&$ 0.067$&$0.083$&$0.733$&&$ 0.012$&$0.064$&$0.956$&&$-0.007$&$0.105$&$0.978$\tabularnewline
				&$-0.320$&$1000$&$24$&&$ 0.034$&$0.058$&$0.893$&&$ 0.001$&$0.059$&$0.957$&&$-0.005$&$0.094$&$0.981$\tabularnewline
				&$-0.320$&$1000$&$48$&&$ 0.015$&$0.047$&$0.936$&&$-0.003$&$0.054$&$0.960$&&$-0.004$&$0.082$&$0.985$\tabularnewline
				&$-0.320$&$4000$&$12$&&$ 0.068$&$0.073$&$0.223$&&$ 0.013$&$0.035$&$0.935$&&$-0.005$&$0.054$&$0.963$\tabularnewline
				&$-0.320$&$4000$&$24$&&$ 0.035$&$0.043$&$0.667$&&$ 0.002$&$0.029$&$0.953$&&$-0.004$&$0.047$&$0.967$\tabularnewline
				&$-0.320$&$4000$&$48$&&$ 0.017$&$0.029$&$0.877$&&$ 0.000$&$0.027$&$0.951$&&$-0.001$&$0.041$&$0.969$\tabularnewline
				\hline
		\end{tabular}\end{center}
	\end{footnotesize}
\end{table}

\begin{table}[!t]
	\begin{footnotesize}
		\caption{Monte Carlo simulation results for $\gamma_0$}
		\label{table-monte-acov}  
		\begin{center}
			\begin{tabular}{lrrrcrrrcrrrcrrr}
				\hline\hline
				\multicolumn{1}{l}{\bfseries }&\multicolumn{3}{c}{\bfseries }&\multicolumn{1}{c}{\bfseries }&\multicolumn{3}{c}{\bfseries ED}&\multicolumn{1}{c}{\bfseries }&\multicolumn{3}{c}{\bfseries HPJ}&\multicolumn{1}{c}{\bfseries }&\multicolumn{3}{c}{\bfseries TOJ}\tabularnewline
				\cline{6-8} \cline{10-12} \cline{14-16}
				\multicolumn{1}{l}{}&\multicolumn{1}{c}{true}&\multicolumn{1}{c}{$N$}&\multicolumn{1}{c}{$T$}&\multicolumn{1}{c}{}&\multicolumn{1}{c}{bias}&\multicolumn{1}{c}{rmse}&\multicolumn{1}{c}{cp}&\multicolumn{1}{c}{}&\multicolumn{1}{c}{bias}&\multicolumn{1}{c}{rmse}&\multicolumn{1}{c}{cp}&\multicolumn{1}{c}{}&\multicolumn{1}{c}{bias}&\multicolumn{1}{c}{rmse}&\multicolumn{1}{c}{cp}\tabularnewline
				\hline
				$\gamma_0$ mean&$1.529$&$ 250$&$12$&&$-0.291$&$0.296$&$0.001$&&$-0.072$&$0.101$&$0.804$&&$ 0.000$&$0.099$&$0.938$\tabularnewline
				&$1.529$&$ 250$&$24$&&$-0.159$&$0.167$&$0.151$&&$-0.025$&$0.063$&$0.920$&&$ 0.002$&$0.073$&$0.945$\tabularnewline
				&$1.529$&$ 250$&$48$&&$-0.082$&$0.094$&$0.596$&&$-0.007$&$0.050$&$0.946$&&$ 0.002$&$0.056$&$0.948$\tabularnewline
				&$1.529$&$1000$&$12$&&$-0.291$&$0.292$&$0.000$&&$-0.073$&$0.081$&$0.453$&&$-0.002$&$0.048$&$0.948$\tabularnewline
				&$1.529$&$1000$&$24$&&$-0.157$&$0.159$&$0.000$&&$-0.023$&$0.037$&$0.871$&&$ 0.005$&$0.036$&$0.949$\tabularnewline
				&$1.529$&$1000$&$48$&&$-0.082$&$0.086$&$0.077$&&$-0.007$&$0.026$&$0.940$&&$ 0.001$&$0.029$&$0.947$\tabularnewline
				&$1.529$&$4000$&$12$&&$-0.291$&$0.292$&$0.000$&&$-0.073$&$0.075$&$0.022$&&$-0.002$&$0.025$&$0.943$\tabularnewline
				&$1.529$&$4000$&$24$&&$-0.158$&$0.158$&$0.000$&&$-0.024$&$0.028$&$0.628$&&$ 0.004$&$0.019$&$0.944$\tabularnewline
				&$1.529$&$4000$&$48$&&$-0.082$&$0.083$&$0.000$&&$-0.007$&$0.015$&$0.906$&&$ 0.002$&$0.015$&$0.945$\tabularnewline
				\hline
				$\gamma_0$ std&$0.668$&$ 250$&$12$&&$ 0.215$&$0.225$&$0.022$&&$ 0.165$&$0.188$&$0.462$&&$ 0.090$&$0.191$&$0.913$\tabularnewline
				&$0.668$&$ 250$&$24$&&$ 0.144$&$0.152$&$0.103$&&$ 0.074$&$0.097$&$0.787$&&$ 0.020$&$0.107$&$0.938$\tabularnewline
				&$0.668$&$ 250$&$48$&&$ 0.087$&$0.095$&$0.384$&&$ 0.029$&$0.054$&$0.912$&&$ 0.002$&$0.067$&$0.936$\tabularnewline
				&$0.668$&$1000$&$12$&&$ 0.216$&$0.219$&$0.000$&&$ 0.165$&$0.171$&$0.010$&&$ 0.090$&$0.124$&$0.811$\tabularnewline
				&$0.668$&$1000$&$24$&&$ 0.147$&$0.149$&$0.000$&&$ 0.077$&$0.084$&$0.244$&&$ 0.024$&$0.060$&$0.928$\tabularnewline
				&$0.668$&$1000$&$48$&&$ 0.087$&$0.090$&$0.003$&&$ 0.029$&$0.037$&$0.765$&&$ 0.002$&$0.034$&$0.943$\tabularnewline
				&$0.668$&$4000$&$12$&&$ 0.217$&$0.217$&$0.000$&&$ 0.165$&$0.167$&$0.000$&&$ 0.091$&$0.100$&$0.401$\tabularnewline
				&$0.668$&$4000$&$24$&&$ 0.146$&$0.147$&$0.000$&&$ 0.076$&$0.077$&$0.000$&&$ 0.022$&$0.035$&$0.875$\tabularnewline
				&$0.668$&$4000$&$48$&&$ 0.088$&$0.088$&$0.000$&&$ 0.029$&$0.031$&$0.264$&&$ 0.002$&$0.017$&$0.949$\tabularnewline
				\hline
				$\gamma_0$ 25\%Q&$1.055$&$ 250$&$12$&&$-0.451$&$0.453$&$0.000$&&$-0.221$&$0.233$&$0.209$&&$-0.085$&$0.150$&$0.923$\tabularnewline
				&$1.055$&$ 250$&$24$&&$-0.272$&$0.277$&$0.001$&&$-0.093$&$0.121$&$0.788$&&$-0.017$&$0.129$&$0.979$\tabularnewline
				&$1.055$&$ 250$&$48$&&$-0.152$&$0.161$&$0.220$&&$-0.033$&$0.081$&$0.942$&&$ 0.000$&$0.118$&$0.986$\tabularnewline
				&$1.055$&$1000$&$12$&&$-0.453$&$0.454$&$0.000$&&$-0.224$&$0.227$&$0.000$&&$-0.089$&$0.109$&$0.742$\tabularnewline
				&$1.055$&$1000$&$24$&&$-0.274$&$0.275$&$0.000$&&$-0.094$&$0.102$&$0.327$&&$-0.018$&$0.066$&$0.958$\tabularnewline
				&$1.055$&$1000$&$48$&&$-0.154$&$0.156$&$0.000$&&$-0.034$&$0.050$&$0.864$&&$ 0.000$&$0.059$&$0.974$\tabularnewline
				&$1.055$&$4000$&$12$&&$-0.454$&$0.454$&$0.000$&&$-0.225$&$0.226$&$0.000$&&$-0.090$&$0.095$&$0.198$\tabularnewline
				&$1.055$&$4000$&$24$&&$-0.274$&$0.275$&$0.000$&&$-0.095$&$0.097$&$0.003$&&$-0.018$&$0.036$&$0.926$\tabularnewline
				&$1.055$&$4000$&$48$&&$-0.154$&$0.155$&$0.000$&&$-0.034$&$0.039$&$0.551$&&$ 0.000$&$0.030$&$0.964$\tabularnewline
				\hline
				$\gamma_0$ 50\%Q&$1.515$&$ 250$&$12$&&$-0.472$&$0.476$&$0.000$&&$-0.182$&$0.206$&$0.549$&&$-0.054$&$0.175$&$0.960$\tabularnewline
				&$1.515$&$ 250$&$24$&&$-0.274$&$0.281$&$0.006$&&$-0.076$&$0.115$&$0.873$&&$-0.015$&$0.152$&$0.976$\tabularnewline
				&$1.515$&$ 250$&$48$&&$-0.150$&$0.161$&$0.286$&&$-0.027$&$0.082$&$0.951$&&$-0.001$&$0.128$&$0.987$\tabularnewline
				&$1.515$&$1000$&$12$&&$-0.473$&$0.474$&$0.000$&&$-0.183$&$0.189$&$0.048$&&$-0.054$&$0.100$&$0.912$\tabularnewline
				&$1.515$&$1000$&$24$&&$-0.274$&$0.276$&$0.000$&&$-0.075$&$0.087$&$0.605$&&$-0.014$&$0.075$&$0.966$\tabularnewline
				&$1.515$&$1000$&$48$&&$-0.151$&$0.154$&$0.000$&&$-0.028$&$0.049$&$0.894$&&$-0.003$&$0.066$&$0.972$\tabularnewline
				&$1.515$&$4000$&$12$&&$-0.474$&$0.474$&$0.000$&&$-0.184$&$0.185$&$0.000$&&$-0.056$&$0.070$&$0.730$\tabularnewline
				&$1.515$&$4000$&$24$&&$-0.275$&$0.275$&$0.000$&&$-0.076$&$0.079$&$0.071$&&$-0.014$&$0.040$&$0.943$\tabularnewline
				&$1.515$&$4000$&$48$&&$-0.152$&$0.152$&$0.000$&&$-0.029$&$0.035$&$0.696$&&$-0.003$&$0.033$&$0.959$\tabularnewline
				\hline
				$\gamma_0$ 75\%Q&$1.982$&$ 250$&$12$&&$-0.325$&$0.337$&$0.093$&&$-0.037$&$0.146$&$0.948$&&$ 0.029$&$0.260$&$0.980$\tabularnewline
				&$1.982$&$ 250$&$24$&&$-0.169$&$0.187$&$0.468$&&$-0.011$&$0.116$&$0.965$&&$ 0.002$&$0.209$&$0.987$\tabularnewline
				&$1.982$&$ 250$&$48$&&$-0.085$&$0.112$&$0.773$&&$-0.003$&$0.100$&$0.966$&&$-0.001$&$0.178$&$0.988$\tabularnewline
				&$1.982$&$1000$&$12$&&$-0.324$&$0.327$&$0.000$&&$-0.036$&$0.079$&$0.925$&&$ 0.029$&$0.132$&$0.964$\tabularnewline
				&$1.982$&$1000$&$24$&&$-0.165$&$0.169$&$0.027$&&$-0.005$&$0.059$&$0.956$&&$ 0.009$&$0.109$&$0.967$\tabularnewline
				&$1.982$&$1000$&$48$&&$-0.084$&$0.092$&$0.377$&&$-0.003$&$0.050$&$0.958$&&$-0.001$&$0.090$&$0.970$\tabularnewline
				&$1.982$&$4000$&$12$&&$-0.324$&$0.324$&$0.000$&&$-0.036$&$0.050$&$0.826$&&$ 0.029$&$0.071$&$0.938$\tabularnewline
				&$1.982$&$4000$&$24$&&$-0.165$&$0.166$&$0.000$&&$-0.006$&$0.031$&$0.944$&&$ 0.009$&$0.055$&$0.957$\tabularnewline
				&$1.982$&$4000$&$48$&&$-0.083$&$0.085$&$0.006$&&$-0.002$&$0.025$&$0.950$&&$ 0.000$&$0.045$&$0.962$\tabularnewline
				\hline
		\end{tabular}\end{center}
	\end{footnotesize}
\end{table}

\begin{table}[!t]
	\begin{footnotesize}
		\caption{Monte Carlo simulation results for $\rho_1$}
		\label{table-monte-acor}   
		\begin{center}
			\begin{tabular}{lrrrcrrrcrrrcrrr}
				\hline\hline
				\multicolumn{1}{l}{\bfseries }&\multicolumn{3}{c}{\bfseries }&\multicolumn{1}{c}{\bfseries }&\multicolumn{3}{c}{\bfseries ED}&\multicolumn{1}{c}{\bfseries }&\multicolumn{3}{c}{\bfseries HPJ}&\multicolumn{1}{c}{\bfseries }&\multicolumn{3}{c}{\bfseries TOJ}\tabularnewline
				\cline{6-8} \cline{10-12} \cline{14-16}
				\multicolumn{1}{l}{}&\multicolumn{1}{c}{true}&\multicolumn{1}{c}{$N$}&\multicolumn{1}{c}{$T$}&\multicolumn{1}{c}{}&\multicolumn{1}{c}{bias}&\multicolumn{1}{c}{rmse}&\multicolumn{1}{c}{cp}&\multicolumn{1}{c}{}&\multicolumn{1}{c}{bias}&\multicolumn{1}{c}{rmse}&\multicolumn{1}{c}{cp}&\multicolumn{1}{c}{}&\multicolumn{1}{c}{bias}&\multicolumn{1}{c}{rmse}&\multicolumn{1}{c}{cp}\tabularnewline
				\hline
				$\rho_1$ mean&$0.397$&$ 250$&$12$&&$-0.199$&$0.200$&$0.000$&&$ 0.008$&$0.030$&$0.939$&&$ 0.004$&$0.054$&$0.949$\tabularnewline
				&$0.397$&$ 250$&$24$&&$-0.097$&$0.098$&$0.000$&&$ 0.005$&$0.021$&$0.944$&&$ 0.004$&$0.033$&$0.951$\tabularnewline
				&$0.397$&$ 250$&$48$&&$-0.047$&$0.049$&$0.087$&&$ 0.002$&$0.016$&$0.945$&&$ 0.001$&$0.021$&$0.948$\tabularnewline
				&$0.397$&$1000$&$12$&&$-0.200$&$0.200$&$0.000$&&$ 0.007$&$0.016$&$0.916$&&$ 0.004$&$0.028$&$0.944$\tabularnewline
				&$0.397$&$1000$&$24$&&$-0.097$&$0.097$&$0.000$&&$ 0.005$&$0.011$&$0.919$&&$ 0.003$&$0.017$&$0.946$\tabularnewline
				&$0.397$&$1000$&$48$&&$-0.047$&$0.048$&$0.000$&&$ 0.002$&$0.008$&$0.939$&&$ 0.000$&$0.011$&$0.946$\tabularnewline
				&$0.397$&$4000$&$12$&&$-0.199$&$0.200$&$0.000$&&$ 0.008$&$0.010$&$0.825$&&$ 0.003$&$0.014$&$0.941$\tabularnewline
				&$0.397$&$4000$&$24$&&$-0.097$&$0.097$&$0.000$&&$ 0.005$&$0.007$&$0.818$&&$ 0.004$&$0.009$&$0.928$\tabularnewline
				&$0.397$&$4000$&$48$&&$-0.047$&$0.048$&$0.000$&&$ 0.002$&$0.005$&$0.907$&&$ 0.000$&$0.005$&$0.941$\tabularnewline
				\hline
				$\rho_1$ std&$0.198$&$ 250$&$12$&&$ 0.109$&$0.109$&$0.000$&&$ 0.045$&$0.050$&$0.496$&&$-0.030$&$0.055$&$0.889$\tabularnewline
				&$0.198$&$ 250$&$24$&&$ 0.058$&$0.059$&$0.000$&&$ 0.007$&$0.018$&$0.931$&&$-0.012$&$0.034$&$0.924$\tabularnewline
				&$0.198$&$ 250$&$48$&&$ 0.029$&$0.031$&$0.128$&&$ 0.001$&$0.013$&$0.945$&&$-0.003$&$0.021$&$0.940$\tabularnewline
				&$0.198$&$1000$&$12$&&$ 0.109$&$0.109$&$0.000$&&$ 0.045$&$0.047$&$0.018$&&$-0.030$&$0.038$&$0.739$\tabularnewline
				&$0.198$&$1000$&$24$&&$ 0.058$&$0.059$&$0.000$&&$ 0.008$&$0.012$&$0.845$&&$-0.011$&$0.019$&$0.885$\tabularnewline
				&$0.198$&$1000$&$48$&&$ 0.030$&$0.030$&$0.000$&&$ 0.001$&$0.006$&$0.950$&&$-0.002$&$0.010$&$0.946$\tabularnewline
				&$0.198$&$4000$&$12$&&$ 0.109$&$0.109$&$0.000$&&$ 0.045$&$0.046$&$0.000$&&$-0.030$&$0.032$&$0.265$\tabularnewline
				&$0.198$&$4000$&$24$&&$ 0.058$&$0.058$&$0.000$&&$ 0.008$&$0.009$&$0.545$&&$-0.011$&$0.014$&$0.697$\tabularnewline
				&$0.198$&$4000$&$48$&&$ 0.030$&$0.030$&$0.000$&&$ 0.001$&$0.003$&$0.940$&&$-0.002$&$0.006$&$0.916$\tabularnewline
				\hline
				$\rho_1$ 25\%Q&$0.263$&$ 250$&$12$&&$-0.275$&$0.277$&$0.000$&&$-0.011$&$0.053$&$0.957$&&$ 0.079$&$0.133$&$0.926$\tabularnewline
				&$0.263$&$ 250$&$24$&&$-0.135$&$0.138$&$0.000$&&$ 0.005$&$0.038$&$0.957$&&$ 0.015$&$0.079$&$0.975$\tabularnewline
				&$0.263$&$ 250$&$48$&&$-0.066$&$0.069$&$0.105$&&$ 0.003$&$0.030$&$0.964$&&$ 0.002$&$0.056$&$0.985$\tabularnewline
				&$0.263$&$1000$&$12$&&$-0.277$&$0.277$&$0.000$&&$-0.013$&$0.029$&$0.926$&&$ 0.079$&$0.096$&$0.721$\tabularnewline
				&$0.263$&$1000$&$24$&&$-0.137$&$0.137$&$0.000$&&$ 0.003$&$0.019$&$0.951$&&$ 0.010$&$0.041$&$0.956$\tabularnewline
				&$0.263$&$1000$&$48$&&$-0.067$&$0.067$&$0.000$&&$ 0.003$&$0.015$&$0.954$&&$ 0.002$&$0.028$&$0.971$\tabularnewline
				&$0.263$&$4000$&$12$&&$-0.277$&$0.277$&$0.000$&&$-0.012$&$0.018$&$0.840$&&$ 0.079$&$0.084$&$0.170$\tabularnewline
				&$0.263$&$4000$&$24$&&$-0.137$&$0.137$&$0.000$&&$ 0.003$&$0.010$&$0.938$&&$ 0.011$&$0.022$&$0.925$\tabularnewline
				&$0.263$&$4000$&$48$&&$-0.067$&$0.067$&$0.000$&&$ 0.003$&$0.008$&$0.939$&&$ 0.002$&$0.014$&$0.959$\tabularnewline
				\hline
				$\rho_1$ 50\%Q&$0.397$&$ 250$&$12$&&$-0.182$&$0.184$&$0.000$&&$ 0.020$&$0.049$&$0.944$&&$-0.020$&$0.098$&$0.967$\tabularnewline
				&$0.397$&$ 250$&$24$&&$-0.085$&$0.088$&$0.019$&&$ 0.013$&$0.036$&$0.940$&&$ 0.007$&$0.068$&$0.972$\tabularnewline
				&$0.397$&$ 250$&$48$&&$-0.040$&$0.045$&$0.407$&&$ 0.005$&$0.027$&$0.958$&&$ 0.001$&$0.049$&$0.982$\tabularnewline
				&$0.397$&$1000$&$12$&&$-0.183$&$0.183$&$0.000$&&$ 0.019$&$0.030$&$0.868$&&$-0.024$&$0.055$&$0.938$\tabularnewline
				&$0.397$&$1000$&$24$&&$-0.086$&$0.086$&$0.000$&&$ 0.011$&$0.020$&$0.903$&&$ 0.004$&$0.034$&$0.962$\tabularnewline
				&$0.397$&$1000$&$48$&&$-0.041$&$0.042$&$0.008$&&$ 0.004$&$0.014$&$0.940$&&$-0.001$&$0.025$&$0.968$\tabularnewline
				&$0.397$&$4000$&$12$&&$-0.183$&$0.183$&$0.000$&&$ 0.019$&$0.022$&$0.608$&&$-0.025$&$0.035$&$0.837$\tabularnewline
				&$0.397$&$4000$&$24$&&$-0.085$&$0.086$&$0.000$&&$ 0.012$&$0.014$&$0.710$&&$ 0.004$&$0.018$&$0.946$\tabularnewline
				&$0.397$&$4000$&$48$&&$-0.041$&$0.041$&$0.000$&&$ 0.004$&$0.008$&$0.900$&&$ 0.000$&$0.013$&$0.951$\tabularnewline
				\hline
				$\rho_1$ 75\%Q&$0.531$&$ 250$&$12$&&$-0.106$&$0.109$&$0.011$&&$ 0.045$&$0.063$&$0.866$&&$ 0.012$&$0.096$&$0.969$\tabularnewline
				&$0.531$&$ 250$&$24$&&$-0.047$&$0.051$&$0.393$&&$ 0.012$&$0.036$&$0.951$&&$-0.003$&$0.068$&$0.980$\tabularnewline
				&$0.531$&$ 250$&$48$&&$-0.021$&$0.028$&$0.782$&&$ 0.004$&$0.028$&$0.961$&&$-0.002$&$0.051$&$0.982$\tabularnewline
				&$0.531$&$1000$&$12$&&$-0.106$&$0.106$&$0.000$&&$ 0.045$&$0.050$&$0.502$&&$ 0.013$&$0.049$&$0.954$\tabularnewline
				&$0.531$&$1000$&$24$&&$-0.046$&$0.047$&$0.009$&&$ 0.014$&$0.022$&$0.879$&&$ 0.000$&$0.034$&$0.968$\tabularnewline
				&$0.531$&$1000$&$48$&&$-0.021$&$0.023$&$0.417$&&$ 0.004$&$0.014$&$0.946$&&$-0.002$&$0.025$&$0.969$\tabularnewline
				&$0.531$&$4000$&$12$&&$-0.105$&$0.105$&$0.000$&&$ 0.045$&$0.046$&$0.021$&&$ 0.013$&$0.027$&$0.920$\tabularnewline
				&$0.531$&$4000$&$24$&&$-0.045$&$0.046$&$0.000$&&$ 0.014$&$0.017$&$0.608$&&$ 0.000$&$0.017$&$0.956$\tabularnewline
				&$0.531$&$4000$&$48$&&$-0.020$&$0.021$&$0.009$&&$ 0.004$&$0.008$&$0.907$&&$-0.001$&$0.013$&$0.964$\tabularnewline
				\hline
		\end{tabular}\end{center}
	\end{footnotesize}
\end{table}

\begin{table}[!t]
	\begin{footnotesize}
		\caption{Monte Carlo simulation results for correlations}
		\label{table-monte-cor}   
		\begin{center}
			\begin{tabular}{lrrrcrrrcrrrcrrr}
				\hline\hline
				\multicolumn{1}{l}{\bfseries }&\multicolumn{3}{c}{\bfseries }&\multicolumn{1}{c}{\bfseries }&\multicolumn{3}{c}{\bfseries ED}&\multicolumn{1}{c}{\bfseries }&\multicolumn{3}{c}{\bfseries HPJ}&\multicolumn{1}{c}{\bfseries }&\multicolumn{3}{c}{\bfseries TOJ}\tabularnewline
				\cline{6-8} \cline{10-12} \cline{14-16}
				\multicolumn{1}{l}{}&\multicolumn{1}{c}{true}&\multicolumn{1}{c}{$N$}&\multicolumn{1}{c}{$T$}&\multicolumn{1}{c}{}&\multicolumn{1}{c}{bias}&\multicolumn{1}{c}{rmse}&\multicolumn{1}{c}{cp}&\multicolumn{1}{c}{}&\multicolumn{1}{c}{bias}&\multicolumn{1}{c}{rmse}&\multicolumn{1}{c}{cp}&\multicolumn{1}{c}{}&\multicolumn{1}{c}{bias}&\multicolumn{1}{c}{rmse}&\multicolumn{1}{c}{cp}\tabularnewline
				\hline
				$\mu$ vs $\gamma_0$&$ 0.193$&$ 250$&$12$&&$-0.142$&$0.156$&$0.377$&&$-0.103$&$0.136$&$0.762$&&$-0.066$&$0.152$&$0.892$\tabularnewline
				&$ 0.193$&$ 250$&$24$&&$-0.098$&$0.117$&$0.638$&&$-0.055$&$0.097$&$0.884$&&$-0.026$&$0.112$&$0.930$\tabularnewline
				&$ 0.193$&$ 250$&$48$&&$-0.061$&$0.088$&$0.826$&&$-0.024$&$0.078$&$0.934$&&$-0.006$&$0.090$&$0.941$\tabularnewline
				&$ 0.193$&$1000$&$12$&&$-0.142$&$0.146$&$0.008$&&$-0.103$&$0.112$&$0.352$&&$-0.068$&$0.095$&$0.799$\tabularnewline
				&$ 0.193$&$1000$&$24$&&$-0.100$&$0.105$&$0.111$&&$-0.057$&$0.070$&$0.692$&&$-0.028$&$0.062$&$0.912$\tabularnewline
				&$ 0.193$&$1000$&$48$&&$-0.062$&$0.070$&$0.485$&&$-0.025$&$0.045$&$0.889$&&$-0.006$&$0.045$&$0.939$\tabularnewline
				&$ 0.193$&$4000$&$12$&&$-0.142$&$0.143$&$0.000$&&$-0.103$&$0.105$&$0.002$&&$-0.068$&$0.075$&$0.467$\tabularnewline
				&$ 0.193$&$4000$&$24$&&$-0.100$&$0.101$&$0.000$&&$-0.056$&$0.060$&$0.200$&&$-0.027$&$0.038$&$0.825$\tabularnewline
				&$ 0.193$&$4000$&$48$&&$-0.062$&$0.064$&$0.021$&&$-0.025$&$0.031$&$0.725$&&$-0.007$&$0.023$&$0.938$\tabularnewline
				\hline
				$\mu$ vs $\rho_1$&$ 0.405$&$ 250$&$12$&&$-0.245$&$0.253$&$0.014$&&$-0.161$&$0.191$&$0.619$&&$-0.087$&$0.200$&$0.911$\tabularnewline
				&$ 0.405$&$ 250$&$24$&&$-0.158$&$0.169$&$0.207$&&$-0.072$&$0.110$&$0.863$&&$-0.017$&$0.134$&$0.940$\tabularnewline
				&$ 0.405$&$ 250$&$48$&&$-0.091$&$0.107$&$0.634$&&$-0.024$&$0.075$&$0.934$&&$ 0.003$&$0.096$&$0.941$\tabularnewline
				&$ 0.405$&$1000$&$12$&&$-0.245$&$0.247$&$0.000$&&$-0.163$&$0.170$&$0.086$&&$-0.092$&$0.128$&$0.823$\tabularnewline
				&$ 0.405$&$1000$&$24$&&$-0.158$&$0.160$&$0.000$&&$-0.070$&$0.082$&$0.604$&&$-0.014$&$0.067$&$0.941$\tabularnewline
				&$ 0.405$&$1000$&$48$&&$-0.091$&$0.095$&$0.095$&&$-0.023$&$0.042$&$0.897$&&$ 0.004$&$0.048$&$0.943$\tabularnewline
				&$ 0.405$&$4000$&$12$&&$-0.245$&$0.246$&$0.000$&&$-0.163$&$0.165$&$0.000$&&$-0.091$&$0.101$&$0.465$\tabularnewline
				&$ 0.405$&$4000$&$24$&&$-0.158$&$0.158$&$0.000$&&$-0.070$&$0.073$&$0.074$&&$-0.014$&$0.036$&$0.927$\tabularnewline
				&$ 0.405$&$4000$&$48$&&$-0.090$&$0.091$&$0.000$&&$-0.023$&$0.029$&$0.745$&&$ 0.004$&$0.024$&$0.946$\tabularnewline
				\hline
				$\gamma_0$ vs $\rho_1$&$-0.286$&$ 250$&$12$&&$ 0.367$&$0.373$&$0.000$&&$ 0.326$&$0.343$&$0.141$&&$ 0.240$&$0.302$&$0.740$\tabularnewline
				&$-0.286$&$ 250$&$24$&&$ 0.271$&$0.280$&$0.019$&&$ 0.175$&$0.200$&$0.565$&&$ 0.076$&$0.172$&$0.917$\tabularnewline
				&$-0.286$&$ 250$&$48$&&$ 0.172$&$0.184$&$0.242$&&$ 0.072$&$0.112$&$0.870$&&$ 0.011$&$0.122$&$0.950$\tabularnewline
				&$-0.286$&$1000$&$12$&&$ 0.367$&$0.369$&$0.000$&&$ 0.327$&$0.332$&$0.000$&&$ 0.241$&$0.258$&$0.247$\tabularnewline
				&$-0.286$&$1000$&$24$&&$ 0.271$&$0.273$&$0.000$&&$ 0.174$&$0.181$&$0.054$&&$ 0.078$&$0.110$&$0.825$\tabularnewline
				&$-0.286$&$1000$&$48$&&$ 0.170$&$0.174$&$0.000$&&$ 0.070$&$0.083$&$0.637$&&$ 0.008$&$0.063$&$0.940$\tabularnewline
				&$-0.286$&$4000$&$12$&&$ 0.368$&$0.368$&$0.000$&&$ 0.328$&$0.329$&$0.000$&&$ 0.242$&$0.246$&$0.000$\tabularnewline
				&$-0.286$&$4000$&$24$&&$ 0.271$&$0.271$&$0.000$&&$ 0.174$&$0.176$&$0.000$&&$ 0.077$&$0.086$&$0.479$\tabularnewline
				&$-0.286$&$4000$&$48$&&$ 0.171$&$0.172$&$0.000$&&$ 0.070$&$0.074$&$0.094$&&$ 0.009$&$0.032$&$0.938$\tabularnewline
				\hline
		\end{tabular}\end{center}
	\end{footnotesize}
\end{table}

The simulation result demonstrates that our asymptotic analyses provide information about the finite-sample behavior and the importance of bias correction.
First, ED has large biases in some parameters of interest, such as the quantities $\gamma_{0,i}$ and $Cor(\gamma_{0,i}, \rho_{1,i})$.
Second, for many parameters, the coverage probabilities of ED differ significantly from 0.95 because of these large biases.
Third, the biases and coverage probabilities of ED can improve when $T$ is large, although significant biases can remain, even with a large $T$.
These results recommend the importance of developing a bias-correction method.

The split-panel jackknife bias correction reduces biases and improves coverage probabilities for many parameters.
HPJ can work well, especially when biases in ED are large.
The coverage probabilities of HPJ are satisfactory for about half of the cases, in particular those with large $T$ and those in which the parameter of interest is $\mu_i$ and $\rho_{1,i}$. 
Conversely, when $T$ is small and when the parameter of interest is $\gamma_{0,i}$ or std, they are not satisfactory. We suspect that large higher-order biases caused by a small $T$ or highly nonlinear parameters may be present in those cases in which HPJ does not work well.
For such cases, TOJ can further improve both the biases and coverage probabilities, which can be expected by our discussion for higher-order jackknife in the supplementary appendix.
In contrast, in some cases TOJ eliminates biases at the inevitable cost of inflation of standard deviations, which may lead to wider confidence intervals, and the coverage probabilities for TOJ may be over 0.95 when estimating some quantiles.

In summary, our recommendation based on these simulation results is to employ split-panel jackknife bias-corrected estimation.
When HPJ and TOJ estimates are close to each other, both estimates could be reliable.
In contrast, when both estimates differ due to a severe higher-order bias, we could rely on TOJ, especially when estimating highly nonlinear parameters, while being cautious about the precision of point estimates.
ED is not recommended.

\section{Conclusion}\label{sec-conclusion}
This paper proposes methods to analyze heterogeneous dynamic structures using panel data.
Our methods are easily implemented without requiring a model specification.
We first compute the sample mean, autocovariances, and autocorrelations for each unit. 
We then use these to estimate the parameters of interest, such as the distribution function, the quantile function, and the other moments of the heterogeneous mean, autocovariances, and/or autocorrelations.
We establish conditions on the relative magnitudes of $N$ and $T$ under which the estimator for the distribution function does not suffer from asymptotic bias.
When the parameter of interest can be written as the expected value of a smooth function of the heterogeneous mean and/or autocovariances, the bias of the estimator is of order $O(1/T)$ and can be reduced using the split-panel jackknife bias correction.
In addition, we develop inference based on the cross-sectional bootstrap and provide an extension based on the proposed procedures involving the testing of differences in heterogeneous dynamic structures across distinct groups.
We apply our procedures to the dynamics of LOP deviations in different US cities for various items and obtain new empirical evidence for significant heterogeneity.
The results of the Monte Carlo simulations demonstrate the desirable properties of the proposed procedures.

\paragraph{Future work.}
Several future research topics are possible.
First, it would be interesting to develop a formal testing procedure to examine whether the dynamics are heterogeneous.
We are currently working on this extension.
Second, it would be interesting to examine quantities in addition to means and autocovariances.
For example, \cite{ArellanoBlundellBonhomme15} highlight the importance of nonlinearity and conditional skewness in earnings and consumption dynamics. 
Third, while we developed our analysis for stationary panel data, it would be interesting to consider nonstationary panel data.
Finally, while we focus only on balanced panel data, an analysis based on unbalanced panel data would be useful.

\appendix

\section{Appendix: Proofs and technical lemmas}\label{sec-appendix}

This appendix presents the proofs of the theorems and the technical lemmas used to prove the theorems.
Section \ref{subsec-proof-theorem} contains the proofs for the theorems in the main text.
The technical lemmas are given in Section \ref{subsec-technical-lemma}.

\subsection{Proofs of theorems} \label{subsec-proof-theorem}

This section contains the proofs of the theorems in the main text. 
We repeatedly cite the results in \citet{vanderVaartWellner96}, subsequently abbreviated as VW.
We also denote a generic constant by $M < \infty$ throughout.

\subsubsection{Proof of Theorem \ref{thm-gc}}

Let $\mathbb{P}_N = \mathbb{P}_N^{\hat \xi}$, $P_T = P_T^{\hat \xi}$, and $P_0 = P_0^{\xi}$ be the probability measures defined in the main body for $\hat \xi_i = \hat \mu_i$, $\hat \gamma_{k,i}$ or $\hat \rho_{k,i}$, and $\xi_i = \mu_i$, $\gamma_{k,i}$ or $\rho_{k,i}$, respectively.
Let $\mathbb{F}_N$, $F_T$ and $F_0$ be the corresponding CDFs.

By the triangle inequality, we have $\sup_{f \in \mathcal{F}} | \mathbb{P}_{N}f - P_{0}f | \leq \sup_{f \in \mathcal{F}} | \mathbb{P}_{N}f - P_Tf | + \sup_{f \in \mathcal{F}} | P_T f - P_{0}f |$.
For the second term on the right-hand side, 
Corollary \ref{lem-second-mu} or Lemma \ref{lem-gamma-moment}, or \ref{lem-rho-moment} for $\hat \xi_i = \hat \mu_i$, $\hat \gamma_{k,i}$, or $\hat \rho_{k,i}$, respectively, implies that $\hat \xi_i$ converges to $\xi_i$ in mean square convergence and thus also implies that $\hat \xi_i$ converges to $\xi_i$ in distribution.
Hence, Lemma 2.11 in \citet{vanderVaart98} implies that $\sup_{f \in \mathcal{F}} | P_T f - P_{0}f | \to 0$ as $\xi_i$ is continuously distributed by Assumption \ref{as-mu-con}.a, \ref{as-gamma-con}.a, or \ref{as-rho-con}.a.

We then show that the first term almost surely converges to $0$.
Note that, for $f=\mathbf{1}_{(-\infty,a]}$, $\mathbb{P}_{N}f = \mathbb{F}_{N}(a)$, and $E(\mathbb{F}_{N}(a)) = \Pr(\hat \xi_i \leq a)=P_T f$.
We first fix a monotonic sequence $T= T(N)$ such that $T\to \infty$ as $N\to \infty$,
which transforms our sample into triangular arrays.
We use the strong law of large numbers for triangular arrays (see, e.g., \citealp[Theorem 2]{Hu1989}).
This is possible because under Assumption \ref{as-basic}, $\mathbf{1}(\hat \xi_i \leq a)$ for any $a \in \mathbb{R}$ is i.i.d. across units, the condition (1.5) in \citet{Hu1989} is clearly satisfied, and the condition (1.6) in \citet{Hu1989} is also satisfied when we set $X=2$ in condition (1.6).
Thus, we have $\mathbb{F}_N(a) - \Pr(\hat \xi_i \leq a) \stackrel{as}{\longrightarrow}  0$ and $\mathbb{F}_N(a-)  -  \Pr(\hat \xi_i < a) \stackrel{as}{\longrightarrow}  0$ for every $a \in \mathbb{R}$, when $T(N) \to \infty$ as $N \to \infty$.
Given a fixed $\varepsilon > 0$, there exists a partition $-\infty = a_0 < a_1 <\cdots< a_L =\infty$ such that $\Pr(\xi_i < a_l) - \Pr(\xi_i \leq a_{l-1}) < \varepsilon /3$ for every $l$.
We showed $	\sup_{f \in \mathcal{F}} \left| P_T f - P_{0}f \right| \to 0$, which implies that for sufficiently large $N,T$, $\sup_{f \in \mathcal{F}} \left| P_T f - P_{0}f \right| < \varepsilon /3$. Therefore, we have 
$\Pr(\hat \xi_i < a_l) - \Pr(\hat \xi_i \leq a_{l-1}) < \varepsilon$ for every $l$.
The rest of the proof is the same as that for Theorem 19.1 in \citet{vanderVaart98}.
For $a_{l-1} \leq a < a_{l}$,
\begin{align*}
	\mathbb{F}_N (a) - \Pr(\hat \xi_i \leq a) &\leq \mathbb{F}_N (a_l-) - \Pr(\hat \xi_i < a_l) + \varepsilon,\\
	\mathbb{F}_N (a) - \Pr(\hat \xi_i \leq a) &\geq \mathbb{F}_N (a_{l-1}-) - \Pr(\hat \xi_i < a_{l-1}) - \varepsilon.
\end{align*}
Accordingly, we have $\limsup_{N,T \to \infty} (\sup_{f \in \mathcal{F}} |\mathbb{P}_{N}f - P_T f|) \leq \varepsilon$ almost surely.
This is true for every $\varepsilon > 0$, and we thus get $\sup_{f \in \mathcal{F}} | \mathbb{P}_{N}f - P_T f | \stackrel{as}{\longrightarrow} 0$.
We note that this result holds for all monotonic diagonal paths $N \to \infty, T(N) \to \infty$. 
As stated in REMARKS (a) in \citet{PhillipsMoon99}, it thus holds under double asymptotics $N,T \to \infty$.
Consequently, we obtain the desired result by the continuous mapping theorem.
\qed

\subsubsection{Proof of Theorem \ref{thm-fclt}}

The proof is based on the decomposition in \eqref{fclt_1} and \eqref{fclt_2}.
To study the asymptotic behavior of \eqref{fclt_1}, we use Lemma 2.8.7 in VW.
We first fix a monotonic sequence $T= T(N)$ such that $T(N) \to \infty$ as $N\to \infty$, making our sample into triangular arrays.
By Theorem 2.8.3, Example 2.5.4, and Example 2.3.4 in VW, the class $\mathcal{F}$ is Donsker and pre-Gaussian uniformly in $\{P_T\}$.
Thus, we need to check the conditions $(2.8.5)$ and $(2.8.6)$ in VW.
The condition $(2.8.6)$ in VW is immediate for the envelope function $F=1$ (constant).

We check the condition $(2.8.5)$ in VW.
Let $\rho_{P_T}$ and $\rho_{P_{0}}$ be the variance semi-metrics with respect to $P_T$ and $P_{0}$, respectively.
Then, 
\begin{align*}
	&\sup_{f,g \in \mathcal{F}} | \rho_{P_T}(f,g) - \rho_{P_{0}}(f,g) | \\
	=& \sup_{f,g \in \mathcal{F}} | \sqrt{P_T((f-g)-P_T(f-g))^{2}} - \sqrt{P_{0}((f-g)-P_{0}(f-g))^{2}} | \\
	=& \sup_{a,a' \in \mathbb{R}} | \sqrt{P_T(\mathbf{1}_{(-\infty, a]} - \mathbf{1}_{(-\infty, a']} - P_T(\mathbf{1}_{(-\infty, a]} -\mathbf{1}_{(-\infty, a']}))^{2}} \\
	& - \sqrt{P_{0}(\mathbf{1}_{(-\infty, a]} - \mathbf{1}_{(-\infty, a']} - P_{0}(\mathbf{1}_{(-\infty, a]} -\mathbf{1}_{(-\infty, a']}))^{2}} | \\
	\leq& \sup_{a,a' \in \mathbb{R}} | P_T(\mathbf{1}_{(-\infty, a]} - \mathbf{1}_{(-\infty, a']} - P_T(\mathbf{1}_{(-\infty, a]} -\mathbf{1}_{(-\infty, a']}))^{2}\\
	& - P_{0}(\mathbf{1}_{(-\infty, a]} - \mathbf{1}_{(-\infty, a']} - P_{0}(\mathbf{1}_{(-\infty, a]} -\mathbf{1}_{(-\infty, a']}))^{2}|^{1/2},
\end{align*}
where the first inequality follows from the triangle inequality.
Without loss of generality, we assume $a >a'$.
Then, by simple algebra, 
\begin{align*}
	& \sup_{f,g \in \mathcal{F}} | \rho_{P_T}(f,g) - \rho_{P_{0}}(f,g) | \\
	\leq & \sup_{a, a' \in \mathbb{R}} \left| (P_T\mathbf{1}_{(-\infty,a]} - P_{0}\mathbf{1}_{(-\infty,a]}) - (P_T\mathbf{1}_{(-\infty,a]}\mathbf{1}_{(-\infty,a']} - P_{0}\mathbf{1}_{(-\infty,a]}\mathbf{1}_{(-\infty,a']})  \right. \\
	& + (P_T\mathbf{1}_{(-\infty,a']} - P_{0}\mathbf{1}_{(-\infty,a']}) -((P_T\mathbf{1}_{(-\infty,a]})^2 - (P_{0}\mathbf{1}_{(-\infty,a]})^2) \\
	& - ((P_T\mathbf{1}_{(-\infty,a']})^2 - (P_{0}\mathbf{1}_{(-\infty,a]'})^2) + 2(P_T\mathbf{1}_{(-\infty,a]}P_T\mathbf{1}_{(-\infty,a']} - P_T\mathbf{1}_{(-\infty,a]}P_{0}\mathbf{1}_{(-\infty,a']})\\
	& \left. + \; 2(P_T\mathbf{1}_{(-\infty,a]}P_{0}\mathbf{1}_{(-\infty,a']} - P_{0}\mathbf{1}_{(-\infty,a]}P_{0}\mathbf{1}_{(-\infty,a']})  \right| ^{1/2}\\
	\leq & \; 11 \sup_{a \in \mathbb{R}} \left|P_T\mathbf{1}_{(-\infty,a]} - P_{0}\mathbf{1}_{(-\infty,a]}\right|^{1/2}\\
	\rightarrow & \; 0,
\end{align*}
where the last conclusion follows from Lemma 2.11 in \citet{vanderVaart98}, and $\hat \xi_i \stackrel{p}{\longrightarrow} \xi_i$, which follows from Corollary \ref{lem-second-mu} or Lemma \ref{lem-gamma-moment}, or \ref{lem-rho-moment} for $\hat \xi_i = \hat \mu_i$, $\hat \gamma_{k,i}$, or $\hat \rho_{k,i}$, respectively.
Therefore, condition (2.8.5) in VW is satisfied.

Therefore, by Lemma 2.8.7 in VW, we show that
\begin{eqnarray}
	\mathbb{G}_{N,P_T} \leadsto \mathbb{G}_{P_{0}} \qquad \mbox{in} \quad\ell^{\infty}(\mathcal{F}). \label{eq-w-conv}
\end{eqnarray}
Note that \eqref{eq-w-conv} holds for all monotonic diagonal paths $T(N) \to \infty$ as $N \to \infty$. 
As in REMARKS (a) of \citet{PhillipsMoon99}, \eqref{eq-w-conv} thus holds under double asymptotics $N,T \to \infty$.
 
Next, we study the asymptotic behavior of \eqref{fclt_2}: $\sqrt{N}(P_Tf-P_{0}f)$.
Because the nonstochastic function sequence $P_Tf-P_{0}f$ is uniformly bounded in $f \in \mathcal{F}$, 
we should consider the convergence rate of $\sup_{f \in \mathcal{F}} |P_Tf - P_{0}f |$.
Lemma \ref{lem-dis-diff-mu}, \ref{lem-dis-diff-gamma}, or \ref{lem-dis-diff-rho} for $\hat \xi_i = \hat \mu_i$, $\hat \gamma_{k,i}$, or $\hat \rho_{k,i}$, respectively, shows that $\sup_{f \in \mathcal{F}} | P_Tf - P_{0}f | = O( T^{-2/(3+\epsilon)})$ for any $\epsilon \in (0, 1/3)$.

Therefore, given $N^{3 + \epsilon}/T^{4} \to 0$ for some $\epsilon \in (0, 1/3)$, the desired result holds by Slutsky's theorem. 
\qed

\subsubsection{Proofs of Theorems \ref{thm-h}-\ref{thm-dist-bootstrap}}
These proofs are included in the supplement.

\subsubsection{Proof of Theorem \ref{thm-KS2}}
We first observe that
\begin{align*}
	KS_2 = \left\| \sqrt{\frac{N_1N_2}{N_1+N_2}}
 (\mathbb{P}_{N_1,(1)} - P_{0,(1)}) - \sqrt{\frac{N_1N_2}{N_1+N_2}} 
(\mathbb{P}_{N_2,(2)} - P_{0,(2)} )+  
\sqrt{\frac{N_1N_2}{N_1+N_2}}(P_{0,(1)} - P_{0,(2)}) \right\|_{\infty}.
\end{align*}
Note that, under Assumption \ref{as-ks}, $\sqrt{N_1}(\mathbb{P}_{N_1,(1)} - P_{0,(1)})$ and $\sqrt{N_2}(\mathbb{P}_{N_2,(2)} -  P_{0,(2)})$ jointly converge in distribution to independent Brownian processes $\mathbb{G}_{P_{0,(1)}}$ and $\mathbb{G}_{P_{0,(2)}}$  given $N_1,T_1\to\infty$ with $N_1^{3+ \epsilon} /T_1^{4} \to 0$ and $N_2,T_2\to\infty$ with $N_2^{3 + \epsilon} /T_2^{4} \to 0$ for some $\epsilon \in (0, 1/3)$ by Theorem \ref{thm-fclt}.
Therefore, under $H_0:P_{0,(1)} = P_{0,(2)}$, $KS_2$ converges in distribution to $\| \sqrt{1-\lambda}\mathbb{G}_{P_{0,(1)}} - \sqrt{\lambda}\mathbb{G}_{P_{0,(2)}} \|_{\infty}$ by the continuous mapping theorem given $N_1 / (N_1 + N_2) \to \lambda \in (0,1)$.
It is easy to see that the distribution of the limit random variable $\sqrt{1-\lambda}\mathbb{G}_{P_{0,(1)}} - \sqrt{\lambda}\mathbb{G}_{P_{0,(2)}}$ is identical to that of $\mathbb{G}_{P_{0,(1)}}$ under $H_0$.
Thus, we have the desired result.
\qed

\subsection{Technical lemmas} \label{subsec-technical-lemma}

\begin{lemma}[\citealp{GalvaoKato14} based on \citealp{Davydov1968}]
\label{lem-mixing-var}
 Let $\{\upsilon_t\}_{t=1}^{\infty}$ denote a stationary process taking values in $\mathbb{R}$ and let $\alpha(m)$ denote its $\alpha$-mixing coefficients.
Suppose that $ E(|\upsilon_1| ^q) <\infty$ and $\sum_{m=1}^\infty \alpha(m)^{1-2/q} <\infty$ for some $q>2$.
Then, we have $var \left( \sum_{t=1}^T \upsilon_t \right)\le C T$ with $C=12 (E(|\upsilon_1|^q))^{2/q} \sum_{m=0}^{\infty} \alpha (m)^{1-2/q}$.
\end{lemma}

\begin{proof}
 The proof is available in \citet{GalvaoKato14} (the discussion after Theorem C.1).
\end{proof}

\begin{lemma}[\citealp{Yokoyama1980}]
\label{lem-yokoyama}
 Let $\{\upsilon_t\}_{t=1}^{\infty}$ denote a strictly stationary $\alpha$-mixing process taking values in $\mathbb{R}$, and let $\alpha(m)$ denote its $\alpha$-mixing coefficients.
Suppose that $E (\upsilon_t) =0$ and for some constants $\delta > 0$ and $r>2$,
$E(|\upsilon_1 | ^{r+\delta}) <\infty$.
If $\sum_{m=0}^{\infty} (m+1)^{r/2 -1} \alpha (m)^{\delta/ (r+\delta)} < \infty$,
then there exists a constant $C$ independent of $T$ such that $E ( |  \sum_{t=1}^T \upsilon_t |^r ) \le C T^{r/2}$.
\end{lemma}

\begin{proof}
	The proof is available in \citet{Yokoyama1980}.
\end{proof}

\begin{lemma}
\label{lem-moment-w}
Let $r$ be an even natural number.
 Suppose that Assumptions \ref{as-basic}, \ref{as-mixing-c}, and \ref{as-w-moment-c} hold for 
$r_m = r$ and $r_d =r$.
Then, it holds that $E((\bar w_i)^r )\le C T^{ - r / 2 }$.

\end{lemma}
\begin{proof}
The proof is included in the supplement.
\end{proof}

Because $\hat \mu_i - \mu_i = \bar y_i - \mu_i = \bar w_i $, we obtain the following result as a corollary.
\begin{corollary}\label{lem-second-mu}
 Let $r$ be an even natural number.
 Suppose that Assumptions \ref{as-basic}, \ref{as-mixing-c}, and \ref{as-w-moment-c} hold for 
$r_m = r$ and $r_d =r$.
Then we have $E( (\hat \mu_i -\mu_i)^{r}) = O(T^{-r/2})$.
\end{corollary}

\begin{lemma}
\label{lem-mixing}
Let $r$ be an even natural number.
Suppose that Assumptions \ref{as-basic} and \ref{as-mixing-c} hold for $r_m = r$.
Then, $\{w_{it} w_{i,t-k}\}_{t=k+1}^{\infty}$ for a fixed $k$ given $\alpha_i$ 
is strictly stationary and $\alpha$-mixing 
and its mixing coefficients $\{\alpha_{k} (m|i)\}_{m=0}^{\infty}$ possess the following properties: there exists a sequence $\{\alpha_k (m)\}_{m=0}^{\infty}$ such that 
for any $i$ and $m$, $\alpha_k (m|i) \le \alpha_k (m)$ and 
$\sum_{m=0}^{\infty} (m+1)^{r/2-1} \alpha_k (m) ^{\delta / (r+\delta)} < \infty$ for some $\delta>0$.
 \end{lemma}

\begin{proof}
The proof is similar to the proof of Theorem 14.1 in \citet{Davidson1994}.
Clearly, for any $i$ and any $0 \leq m < k$, $\alpha_k (m|i) \leq 1$, and that for any $i$ and any $m \geq k$, $\alpha_k (m|i) \le \alpha (m-k |i) \leq \alpha(m-k)$ by the definition of $\alpha$-mixing coefficients and Assumption \ref{as-mixing-c}.
Thus, we have $\sum_{m=0}^{\infty} (m+1)^{r/2-1} \alpha_k (m)^{\delta / (r+\delta)} \le \sum_{m=0}^{k-1} (m+1)^{r/2-1} + \sum_{m=k}^{\infty} (m+1)^{r/2-1} \alpha (m-k)^{\delta / (r+\delta)} < \infty$ under Assumption \ref{as-mixing-c}.
\end{proof}

\begin{lemma}
\label{lem-moment-wk}
Let $r$ be an even natural number.
 Suppose that Assumptions \ref{as-basic}, \ref{as-mixing-c}, and \ref{as-w-moment-c} hold for $r_m = r$ and $r_d = 2r$.
Then, it holds that $E((\sum_{t=k+1}^T (w_{it}w_{i,t-k} - \gamma_{k,i}))^r) \leq C T^{ r / 2 }$ for some constant $C$.
\end{lemma}

\begin{proof}
 In view of Lemma \ref{lem-mixing}, the lemma follows the same line as that for Lemma \ref{lem-moment-w}.
\end{proof}

\begin{lemma}
\label{lem-gamma-moment}
Let $r$ be an even natural number.
Suppose that Assumptions \ref{as-basic}, \ref{as-mixing-c}, and \ref{as-w-moment-c} hold  for $r_m = 2r$ and $r_d=2r$. Then, we have $E((\hat\gamma_{k,i}-\gamma_{k,i})^{r}) = O(T^{-r/2})$.
\end{lemma}

\begin{proof}
The proof is included in the supplement.
\end{proof}

\begin{lemma}
\label{lem-rho-moment}
Let $r$ be an even natural number.
Suppose that Assumptions \ref{as-basic}, \ref{as-mixing-c}, and \ref{as-w-moment-c} hold for $r_m = 2r$ and $r_d=2r$ and that $\gamma_{0,i} > \epsilon$ almost surely for some constant $\epsilon>0$. We have $E((\hat\rho_{k,i}-\rho_{k,i})^{r}) = O(T^{-r/2})$.
\end{lemma}

\begin{proof}
We observe that $E(( \hat \rho_{k,i} - \rho_{k,i})^r ) = E(( \gamma_{0,i}^{-1} ( \hat \gamma_{k,i} - \gamma_{k,i} ) - \gamma_{0,i}^{-1} \hat \rho_{k,i} ( \hat \gamma_{0,i} - \gamma_{0,i} ) )^r)$.
By Lo\'eve's $c_r$ inequality, we only need to examine the $r$-order moment of each term in parentheses on the right-hand side.
We have $E( ( \gamma_{0,i}^{-1} ( \hat \gamma_{k,i} - \gamma_{k,i}) )^r ) \le M E((\hat\gamma_{k,i}-\gamma_{k,i})^{r})$ for some $M < \infty$ by the assumption that $\gamma_{0,i} >\epsilon$.
Lemma \ref{lem-gamma-moment} implies that $E((\hat\gamma_{k,i}-\gamma_{k,i})^{r}) = T^{-r/2}$.
For the second term, it holds that $E ( ( \gamma_{0,i}^{-1} \hat \rho_{k,i} ( \hat \gamma_{0,i} - \gamma_{0,i}))^r)\le M E((\hat\gamma_{0,i}-\gamma_{0,i})^{r})$ for some $M < \infty$ by the assumption that $\gamma_{0,i} >\epsilon$ and the fact that $\vert \hat \rho_{k,i} \vert \le 1$.
Lemma \ref{lem-gamma-moment} implies that $E((\hat\gamma_{0,i}-\gamma_{0,i})^{r}) = T^{-r/2}$.
We thus have the desired result.
\end{proof}

\begin{lemma}\label{lem-dis-diff-mu}
	Suppose that Assumptions \ref{as-basic}, \ref{as-mixing-c}, \ref{as-w-moment-c}, and \ref{as-mu-con} hold for $r_m = 4$ and $r_d = 4$.
	Let $P_T = P_T^{\hat \mu}$ and $P_0 = P_0^{\mu}$ be the probability measures of $\hat \mu_i$ and $\mu_i$, respectively.
	It holds that $\sup_{f \in \mathcal{F}} | P_T f - P_0 f | = O ( T^{-2 / (3 + 2 \epsilon)} )$ for any $\epsilon \in (0, 1/3)$.
\end{lemma}

\begin{proof}
	The proof is based on the comparison between the characteristic functions of $\hat \mu_j$ and $\mu_j$. 
	In the proof, we use the index $j$ instead of the index $i$ because $i$ is reserved for the imaginary number.
	We introduce the sum of $\hat \mu_j$ and a Gaussian noise to guarantee that terms in the expansion of the characteristic function below are integrable.
	Consider $\hat \mu_j = \mu_j + \bar w_j$, $\tilde \mu_j \coloneqq \hat \mu_j + z = \mu_j + \bar w_j + z$, and $\check \mu_j \coloneqq \mu_j + z$ where $z \sim \mathcal{N}(0,\sigma^2)$ for some $\sigma^2 > 0$ and $z$ is independent of $(\alpha_j, \{ w_{jt} \}_{t=1}^T)$. 
	Below we consider a situation where $\sigma^2 \to 0$ depending on $T \to \infty$.
	Let $P_T$, $\tilde P_T$, $\check P$, and $P_0$ be the probability measures of $\hat \mu_{j}$, $\tilde \mu_j$, $\check \mu_j$, and $\mu_{j}$, respectively.
	We observe that 
	\begin{align} \label{eq:difference}
		\sup_{f \in \mathcal{F}} | P_T f  - P_0 f | \le 
		\sup_{f \in \mathcal{F}} | P_T f  - \tilde P_T f | + \sup_{f \in \mathcal{F}} | \tilde P_T f  - \check P f | + \sup_{f \in \mathcal{F}} | \check P f - P_0 f |.
	\end{align}
	We examine each term on the right-hand side.
	For $f = \mathbf{1}_{(-\infty, a]}$, we write the CDFs $F_T(a) = P_T f = \Pr(\hat \mu_j \le a)$, $\tilde F_T(a) = \tilde P_T f = \Pr(\tilde \mu_j \le a)$, $\check F(a) = \check P f = \Pr(\check \mu_j \le a)$, and $F_0(a) = P_0 f = \Pr(\mu_j \le a)$.
	
	We first examine the first term in \eqref{eq:difference}.
	For any $a \in \mathbb{R}$, we observe that $\tilde F_T (a ) - F_T (a) = E [ E( \mathbf{1} \{ \hat \mu_j + z \le a \}| z ) ] - F_T (a) = E [ F_T (a - z)] - F_T (a)$ because of the law of iterated expectations and the independence between $z$ and $\hat \mu_j$.
	We consider the third-order Taylor expansion of $F_T (a-z)$ around $z= 0$: $F_T ( a-z) = F_T ( a) - z F_T' (a) + z^2 F_T^{''} (a) / 2 - z^3 F_T^{'''} (\tilde a) / 3!$ where $\tilde a$ is between $a-z$ and $a$.
	Noting that $E(z) = 0$, $E(z^2) = \sigma^2$, and $E|z|^3 = O(\sigma^3)$, we obtain that $\left| \tilde F_T (a ) - F_T (a) \right| = \left| E [ F_T (a - z) ] - F_T (a) \right| = O(\sigma^2)$ uniformly over $a \in \mathbb{R}$ by Assumption \ref{as-mu-con}.c. 
	Hence, we have $\sup_{f \in \mathcal{F}} | P_T f  - \tilde P_T f | = O(\sigma^2)$.
	
	We then examine the third term in \eqref{eq:difference}.
	The proof is the same as for the first term.
	Given $z$ is independent of $\mu_j$, we observe that $\check F (a ) - F_0 (a) =  E [ F_0 (a - z)] - F_0 (a)$.
	The third-order Taylor expansion of $F_0 (a-z)$ around $z= 0$ is $F_0 ( a-z) = F_0 ( a) - z F_0' (a) + z^2 F_0^{''} (a) / 2 - z^3 F_0^{'''} (\tilde a) / 3!$.
	By Assumption \ref{as-mu-con}.b, we obtain that $\left| \check F (a ) - F_0 (a) \right| = \left| E [ F_0 (a - z) ] - F_0 (a) \right| = O(\sigma^2)$ uniformly over $a \in \mathbb{R}$. 
	Hence, we have $\sup_{f \in \mathcal{F}} | \check P f  - P_0 f | = O(\sigma^2)$.
	
	Next, we evaluate the second term in \eqref{eq:difference}.
	To this end, we first expand the characteristic functions of $\check \mu_j$ and $\tilde \mu_j$.
	By the independence between $z$ and $(\alpha_j, \{w_{jt}\}_{t=1}^T)$, we have
	\begin{align*}
	\psi_{\check \mu} (\zeta) 
	&\coloneqq E[ \exp (i \zeta \check \mu_j )] 
	= E[\exp (i \zeta z)] E[ \exp (i \zeta \mu_j )]
	= \exp \left( -\frac{1}{2} \sigma^2 \zeta^2 \right) \psi_{\mu}(\zeta), \\
	\psi_{\tilde \mu}(\zeta) 
	&\coloneqq E[ \exp (i \zeta \tilde \mu_j )]
	= E[\exp (i \zeta z)] E[ \exp (i \zeta \hat \mu_j )]
	= \exp \left( -\frac{1}{2} \sigma^2 \zeta^2 \right) \psi_{\hat \mu}(\zeta),
	\end{align*}
	where $\psi_{\mu}(\zeta) \coloneqq E[ \exp (i \zeta \mu_j) ]$ and $\psi_{\hat \mu}(\zeta) \coloneqq E[ \exp (i \zeta \hat \mu_j) ]$ are the characteristic functions of $\mu_j$ and $\hat \mu_j$.
	For the characteristic function of $\hat \mu_j$, we observe that $\psi_{\hat \mu}(\zeta) = E[ \exp (i \zeta \hat \mu_j ) ] = E [ \exp (i \zeta \mu_j ) \exp ( i \zeta \bar w_j ) ]$.
	By Taylor's theorem, it holds that $\exp ( i \zeta  \bar w_j ) = 1 + i \zeta \bar w_j - \zeta^2 (\bar w_j)^2 / 2 - i \zeta^3 ( \bar w_j )^3 \exp (i \zeta \tilde w_j) / 3!$ where $\tilde w_j$ is between 0 and $\bar w_j$.\footnote{Strictly speaking, 
	this expansion may not hold as the mean value expression of the remainder term for Taylor's theorem for complex functions may not exist.
	However, as $\exp(i \zeta \bar w_j) = \sin(\zeta \bar w_j) + i \cos(\zeta \bar w_j)$, applying Taylor's theorem for real functions to $\cos:\mathbb{R} \to [-1, 1]$ and $\sin: \mathbb{R} \to [-1, 1]$ leads to
	\begin{align*}
	\exp ( i \zeta  \bar w_j )
	= 1 + i \zeta \bar w_j - \frac{1}{2} \zeta^2 (\bar w_j)^2 
	- \frac{1}{6} \zeta^3 ( \bar w_j )^3 \big( \cos(c_1) - i \sin(c_2) \big),
	\end{align*}
	where $c_1$ and $c_2$ are located between 0 and $\zeta \bar w_j$ but $c_1 \neq c_2$ in general. 
	Given $\cos(\cdot)$ and $\sin(\cdot)$ are bounded functions, we can obtain the same result in the main body based on this observation.
	}
	Therefore, it holds  that
	\begin{align*}
	E[ \exp (i \zeta \hat \mu_j ) ]
	=& E[ \exp (i \zeta \mu_j) ] +	i \zeta E[ \bar w_j \exp (i \zeta \mu_j) ] \\
	&- \frac{1}{2} \zeta^2 E[ (\bar w_j)^2  \exp (i \zeta \mu_j) ]
	- \frac{1}{3!} i \zeta^3 E[ ( \bar w_j )^3 \exp (i \zeta \tilde w_j) \exp (i \zeta \mu_j) ]\\
	=& \psi_{\mu}(\zeta) - \frac{1}{2} \zeta^2 E[ (\bar w_j)^2  \exp (i \zeta \mu_j) ]
	- \frac{1}{3!} i \zeta^3 E[ ( \bar w_j )^3 \exp (i \zeta \tilde w_j) \exp (i \zeta \mu_j) ],
	\end{align*} 
	where $E[ \bar w_j \exp (i \zeta \mu_j) ] =0$ follows from $E(\bar w_j | j) = 0$.
	Hence, it holds that 
	\begin{align*}
	\psi_{\tilde \mu}(\zeta) = \exp \left( -\frac{1}{2} \sigma^2 \zeta^2 \right) \left( \psi_{\mu}(\zeta) - \frac{1}{2} \zeta^2 E[ (\bar w_j)^2  \exp (i \zeta \mu_j) ] - \frac{1}{3!} i \zeta^3 E[ ( \bar w_j )^3 \exp (i \zeta \tilde w_j) \exp (i \zeta \mu_j) ] \right).
	\end{align*}
	
	We use the inversion theorem (\citealp{gil1951note} and \citealp{wendel1961non}) to bound $\check F (a) - \tilde F_T(a)$ for any $a \in \mathbb{R}$.
	\begin{equation}\label{eq:inverse}
	\begin{split}
	\check F_T(a) - \tilde F (a)
	=& \left( \frac{1}{2} + \frac{1}{2 \pi} \int_{-\infty}^{\infty} \frac{e^{i\zeta a} \psi_{\check \mu}(-\zeta) - e^{-i \zeta a} \psi_{\check \mu}(\zeta)}{i \zeta} d\zeta \right) \\
	& - \left(\frac{1}{2} + \frac{1}{2 \pi} \int_{-\infty}^{\infty} \frac{e^{i\zeta a} \psi_{\tilde \mu}(-\zeta) - e^{-i \zeta a} \psi_{\tilde \mu}(\zeta)}{i \zeta} d\zeta \right) \\
	=& \frac{1}{ \pi }
	\int_{-\infty}^{\infty}
	\frac{e^{-i \zeta a}}{i \zeta} [\psi_{\tilde \mu}(\zeta) - \psi_{\check \mu}(\zeta)] d\zeta \\
	=&  \frac{1}{\pi}
	\int_{-\infty}^{\infty}
	\frac{e^{-i \zeta a}}{i \zeta} \exp \left( -\frac{1}{2} \sigma^2 \zeta^2 \right) \left( 
	- \frac{1}{2} \zeta^2 E \left[ (\bar w_j)^2  \exp (i \zeta \mu_j) \right]
	\right) d \zeta \\
	& +  \frac{1}{\pi}
	\int_{-\infty}^{\infty}
	\frac{e^{-i \zeta a}}{i\zeta} \exp \left( -\frac{1}{2} \sigma^2 \zeta^2 \right) \left( - \frac{1}{3!} i \zeta^3 E \left[ ( \bar w_j )^3 \exp (i \zeta \tilde w_j) \exp (i \zeta \mu_j) \right] \right) d \zeta. 
	\end{split}
	\end{equation}
	We examine each of the two terms.
	
	First, we consider the first term in the last line of \eqref{eq:inverse}.
	\begin{equation} \label{eq:integral}
	\begin{split}
	&
	\frac{1}{\pi} \int_{-\infty}^{\infty} \frac{e^{-i \zeta a} }{i \zeta} \exp \left( -\frac{1}{2} \sigma^2 \zeta^2 \right) \left( - \frac{1}{2} \zeta^2 E \left[ (\bar w_j)^2  \exp (i \zeta \mu_j) \right) \right] d \zeta 
	\\
	= & E \left[ (\bar w_j)^2 \frac{i}{2 \pi} \int_{-\infty}^{\infty} \exp(i \zeta (\mu_j-a )) \zeta \exp \left( -\frac{1}{2} \sigma^2 \zeta^2 \right)  d\zeta \right],
	\end{split}
	\end{equation}
	by Fubini's theorem.
	Here, it holds that
	\begin{align*}
	\frac{i}{2 \pi} \int_{-\infty}^{\infty} \exp(i \zeta (\mu_j-a )) \zeta \exp \left( -\frac{1}{2} \sigma^2 \zeta^2 \right)  d\zeta
	= \frac{a - \mu_j}{\sqrt{2 \pi} \sigma^3} \exp \left( - \frac{(a - \mu_j)^2}{2 \sigma^2} \right).
	\end{align*}
	Thus, equation \eqref{eq:integral} can be written as
	\begin{align*}
	\text{Equation \eqref{eq:integral}} 
	= E\left[ (\bar w_j)^2 \frac{a - \mu_j}{\sqrt{2 \pi} \sigma^3} \exp \left( - \frac{(a - \mu_j)^2}{2 \sigma^2} \right) \right].
	\end{align*}
	To proceed, we define the shorthand notation $Z_j \coloneqq a - \mu_j$.
	We consider any nonrandom $A_\sigma > 0$ that satisfies $A _\sigma \to 0$ as $\sigma \to 0$.
	We then have
	\begin{equation}\label{eq:splitA}
	\begin{split}
		\left| E\left[ (\bar w_j)^2 \frac{Z_j}{\sigma^3} \exp \left( - \frac{Z_j^2}{2 \sigma^2} \right) \right] \right|
		\le & \frac{1}{\sigma^3}E\left[ (\bar w_j)^2 |Z_j| \exp \left( - \frac{Z_j^2}{2 \sigma^2} \right) \right] \\
		=& \frac{1}{\sigma^3}E\left[ (\bar w_j)^2 |Z_j| \exp \left( - \frac{Z_j^2}{2 \sigma^2} \right) \mathbf{1}(|Z_j| \le A_{\sigma}) \right] \\
		& + \frac{1}{\sigma^3}E\left[ (\bar w_j)^2 |Z_j| \exp \left( - \frac{Z_j^2}{2 \sigma^2} \right) \mathbf{1}(|Z_j| > A_{\sigma}) \right]. \\
	\end{split}
	\end{equation}
	For the first term in the last line of \eqref{eq:splitA}, it holds that
	\begin{align*}
	\frac{1}{\sigma^3} E\left[ (\bar w_j)^2 |Z_j| \exp \left( - \frac{Z_j^2}{2 \sigma^2} \right) \mathbf{1}(|Z_j| \le A_{\sigma}) \right]
	\le&  \frac{A_{\sigma}}{\sigma^3} \exp(0) E \left[ (\bar w_j)^2 \mathbf{1}(|Z_j| \le A_{\sigma}) \right] \\
	=&  \frac{A_{\sigma}}{\sigma^3} E \left[ E\left[ (\bar w_j)^2 | \mu_j \right] \mathbf{1}(|Z_j| \le A_{\sigma}) \right] \\
	\le&  \frac{A_{\sigma}}{\sigma^3} \frac{M}{T} E \left[ \mathbf{1}(|Z_j| \le A_{\sigma}) \right] \\
	=& O\left( \frac{A_{\sigma}^{2}}{\sigma^3 T} \right),
	\end{align*}
	where the second inequality follows by Assumption \ref{as-mu-con}.d and the last equality holds by $E[\mathbf{1}(|Z_j| \le A_{\sigma})] = \Pr(|Z_j| \le A_{\sigma}) = \int_0^{A_{\sigma}} f_Z(z) dz = O(A_{\sigma})$ based on the bounded density of $\mu_j$.
	For the second term in the last line of \eqref{eq:splitA}, we have 
	\begin{align*}
	& \frac{1}{\sigma^3}E\left[ (\bar w_j)^2 |Z_j| \exp \left( - \frac{Z_j^2}{2 \sigma^2} \right) \mathbf{1}(|Z_j| > A_{\sigma}) \right] \\
	& \le \frac{1}{\sigma^3} \sqrt{E\left[ (\bar w_j)^4 \right]} \sqrt{E\left[ Z_j^2 \exp^2 \left( - \frac{Z_j^2}{2 \sigma^2} \right) \mathbf{1}(|Z_j| > A_{\sigma}) \right]}\\
	& \le \frac{1}{\sigma^3} \sqrt{E\left[ (\bar w_j)^4 \right]} \sqrt{E(Z_j^2) \exp^2 \left( - \frac{A_{\sigma}^2}{2 \sigma^2} \right) }\\
	&= O\left( \frac{1}{\sigma^3} \right) \cdot O\left( \frac{1}{T} \right) \cdot O\left( \exp \left( - \frac{A_{\sigma}^2}{\sigma^2} \right)  \right) 
	= O\left( \frac{1}{\sigma^3 T} \exp \left( - \frac{A_{\sigma}^2}{\sigma^2} \right)  \right),
	\end{align*}
	where we used the Cauchy--Schwarz inequality, Lemma \ref{lem-moment-w}, and $E(Z_j^2) = O(1)$.
	Note that, if $A_{\sigma} = \sigma^{1 - \epsilon'}$ for any $0< \epsilon' < 1$, it holds that $\exp(-A_{\sigma}^2 / \sigma^2) / A_{\sigma}^{2} = o(1)$ as $\sigma \to 0$.
	Therefore, by setting $A_{\sigma} = \sigma^{1 - \epsilon'}$ for any $0< \epsilon' < 1$, we obtain that
	\begin{align*}
	\text{Equation \eqref{eq:integral}} 
	= O\left( \frac{A_{\sigma}^{2}}{\sigma^3 T} \right) + O\left( \frac{1}{\sigma^3 T} \exp \left( - \frac{A_{\sigma}^2}{\sigma^2} \right)  \right)
	= O\left( \frac{A_{\sigma}^{2}}{\sigma^3 T} \right)
	= O\left( \frac{1}{\sigma^{1 + 2 \epsilon'} T} \right).
	\end{align*}
	
	Next, we consider the second term in the last line of \eqref{eq:inverse}.
	\begin{align*}
	&	\frac{1}{\pi}
	\int_{-\infty}^{\infty}
	\frac{e^{-i \zeta a} }{i\zeta} \exp \left( -\frac{1}{2} \sigma^2 \zeta^2 \right) \left( 
	- \frac{1}{3!} i \zeta^3 E\left[ ( \bar w_j )^3 \exp (i \zeta \tilde w_j) \exp (i \zeta \mu_j) \right] \right) d \zeta \\
	& = -\frac{1}{ 3! \pi} E \left[ (\bar w_j)^3 \int_{-\infty}^{\infty} \exp(i \zeta (\mu_j + \tilde w_j -a )) \zeta^2 \exp \left( -\frac{1}{2} \sigma^2 \zeta^2 \right)  d\zeta \right].
	\end{align*}
	Note that $|\exp(i \zeta (\mu_j + \tilde w_j - a))| \le 1$ and $E |\bar w_j|^3 = O(T^{-3/2})$ by Lemma \ref{lem-moment-w} and H\"older's inequality.
	Given $(2 \pi)^{-1/2} \sigma \exp (- \sigma^2 \zeta^2 /2) $ is the density function of $\mathcal{N}(0, 1/ \sigma^2)$, we thus obtain that
	\begin{align*}
	\left| \frac{1}{ 3! \pi} E \left[ (\bar w_j)^3 \int_{-\infty}^{\infty} \exp(i \zeta (\mu_j + \tilde w_j -a )) \zeta^2 \exp \left( -\frac{1}{2} \sigma^2 \zeta^2 \right)  d\zeta \right] \right|
	= O \left( \frac{1}{\sigma^3 T^{3/2} }\right).
	\end{align*}
	
	In sum, we have shown the order of the second term in \eqref{eq:difference}:
	\begin{align*}
	\left| \tilde F_T(a) - \check F (a) \right|
	= O\left( \frac{1}{\sigma^{1 + 2 \epsilon'} T} + \frac{1}{\sigma^3 T^{3/2} }\right),
	\end{align*}
	uniformly over $a \in \mathbb{R}$ for any $0< \epsilon' < 1$.
		
	Based on the above results, we have shown that 
	\begin{align*}
	\sup_{f \in \mathcal{F}} \left| P_T f - P_0 f \right|
	= O\left( \sigma^2 +  \frac{1}{\sigma^{1 + 2 \epsilon'} T} + \frac{1}{\sigma^3 T^{3/2} }\right).
	\end{align*}
	When we set $\sigma = 1 / T^{1 / (3 + 2\epsilon')}$ with any $0 < \epsilon' < 1 / 6$, we obtain the following convergence result:
	\begin{align*}
	\sup_{f \in \mathcal{F}} \left| P_T f - P_0 f \right|
	= O\left( \sigma^2 +  \frac{1}{\sigma^{1 + 2 \epsilon'} T} \right) 
	= O \left( \frac{1}{T^{2 / (3 + \epsilon)}} \right),
	\end{align*}
	for any $0 < \epsilon = 2\epsilon' < 1/3$.
\end{proof}

\begin{lemma}\label{lem-dis-diff-gamma}
	Suppose that Assumptions \ref{as-basic}, \ref{as-mixing-c}, \ref{as-w-moment-c}, and \ref{as-gamma-con} hold for $r_m = 8$ and $r_d = 8$.
	Let $P_T = P_T^{\hat \gamma_k}$ and $P_0 = P_0^{\gamma_k}$ be the probability measures of $\hat \gamma_{k,i}$ and $\gamma_{k,i}$, respectively.
	It holds that $\sup_{f \in \mathcal{F}} | P_T f - P_0 f | = O ( T^{-2 / (3 + \epsilon)} )$ for any $\epsilon \in (0, 1/3)$.
\end{lemma}

\begin{proof}
The proof is included in the supplement.
\end{proof}

\begin{lemma}\label{lem-dis-diff-rho}
	Suppose that Assumptions \ref{as-basic}, \ref{as-mixing-c}, \ref{as-w-moment-c}, and \ref{as-rho-con} hold for $r_m = 8$ and $r_d = 8$.
	Let $P_T = P_T^{\hat \rho_k}$ and $P_0 = P_0^{\rho_k}$ be the probability measures of $\hat \rho_{k,i}$ and $\rho_{k,i}$, respectively.
	It holds that $\sup_{f \in \mathcal{F}} | P_T f - P_0 f | = O ( T^{-2 / (3 + \epsilon)})$ for any $\epsilon \in (0, 1/3)$.
\end{lemma}

\begin{proof}
The proof is included in the supplement.
\end{proof}

\section*{Acknowledgments}
The authors greatly appreciate the assistance of Yuya Sasaki and Mototsugu Shintani in providing the panel data for Chilean firms and US prices, respectively, and for their helpful discussions.
The authors would also like to thank the editor (Oliver Linton) and the associate editor, three anonymous referees, Debopam Bhattacharya, Richard Blundell, Maurice Bun, Yoosoon Chang, Andrew Chesher, Yasunori Fujikoshi, Kazuhiko Hayakawa, Toshio Honda, Hidehiko Ichimura, Koen Jochmans, Art\=uras Juodis, Hiroaki Kaido, Hiroyuki Kasahara, Kengo Kato, Toru Kitagawa, Yuichi Kitamura, Frank Kleibergen, Roger Koenker, Eiji Kurozumi, Simon Lee, Alexei Onatski, Taisuke Otsu, Hashem Pesaran, Katsumi Shimotsu, Richard Smith, M\r ans S\"oderbom, Liangjun Su, Martin Weidner, Yoon-Jae Whang, Takashi Yamagata, Yohei Yamamoto, and seminar participants at various institutes for their helpful comments and useful discussions. 
Part of this research was conducted while Okui was at Vrije Universiteit Amsterdam and Kyoto University and while Yanagi was at Hitotsubashi University.
All remaining errors are ours. 
Okui acknowledges financial support from the Japan Society for the Promotion of Science (JSPS) under KAKENHI Grant Nos. 22330067, 25780151, 25285067, 15H03329, and 16K03598. 
Yanagi recognizes the financial support of Grant-in-Aid for JSPS Fellows No. 252035 and KAKENHI Grant Nos. 15H06214 and 17K13715.


\bibliographystyle{abbrvnat}
\bibliography{panhet.bib}

\end{document}